\documentclass[12pt]{article} %***
 \usepackage[sectionbib]{natbib}
 \usepackage{array,epsfig,fancyheadings,rotating}
 \RequirePackage{algorithm,algpseudocode}
 \usepackage{algpseudocode}
 \usepackage{subfigure}
 \usepackage{matlab-prettifier}
 \usepackage[]{hyperref}  %<----modified by Ivan
 \usepackage{sectsty, secdot}
 \sectionfont{\fontsize{12}{14pt plus.8pt minus .6pt}\selectfont}
 \renewcommand{\theequation}{\thesection\arabic{equation}}
 \subsectionfont{\fontsize{12}{14pt plus.8pt minus .6pt}\selectfont}
 \usepackage{amsmath}
 \usepackage{amssymb}
 \usepackage{amsfonts}
 \usepackage{multirow}
 \usepackage{amsthm}
 \usepackage{dsfont}
 \usepackage{changes}
 \RequirePackage{mathrsfs,morefloats,booktabs,pdfpages,color,array,hhline,epstopdf,algorithm}
\usepackage{indentfirst}
\newtheorem{thm}{Theorem}
\newtheorem{defin}{Definition}
\newtheorem{lem}{Lemma}
\newtheorem{assum}{Assumption}
\newtheorem{rem}{Remark}

\newtheorem{con}{Condition}
\newtheorem{prop}{Proposition}
\newtheorem{Ex}{Example}
\usepackage[margin=1.0in]{geometry} % or \usepackage[left=1.0in, right=1.0in, top=1.0in, bottom=1.0in]{geometry}
\begin{document}
	\renewcommand{\baselinestretch}{2}
	
	\markright{ \hbox{\footnotesize\rm Statistica Sinica
			%{\footnotesize\bf 24} (2020), 000-000
		}\hfill\\[-13pt]
		\hbox{\footnotesize\rm
			%\href{http://dx.doi.org/10.5705/ss.20??.???}{doi:http://dx.doi.org/10.5705/ss.20??.???}
		}\hfill }
	
	\markboth{\hfill{\footnotesize\rm Huan Qing} \hfill}
	{\hfill {\footnotesize\rm Multi-layer grade of membership model} \hfill}

	$\ $\par
	
	%%%%%%%%%%%%%%%%%%%%%%%%%%%%%%%%%%%%%%%%%%%%%%%%%%%%%%%%%%%%%%%%%%%%%%%%%%%%%%%%%%%%%%%%%%%%%%%%%%%%%%%%%%%%%%%%%%%%%%%%%%%%
	
	\fontsize{12}{14pt plus.8pt minus .6pt}\selectfont \vspace{0.8pc}
	\centerline{\large\bf Grade of membership analysis for multi-layer ordinal categorical data}
	\vspace{.4cm}
	\centerline{Huan Qing}
	\vspace{.4cm}
	\centerline{\it School of Economics and Finance, Chongqing University of Technology}

	\vspace{.55cm} \fontsize{9}{11.5pt plus.8pt minus.6pt}\selectfont
	
	%%%%%%%%%%%%%%%%%%%%%%%%%%%%%%%%%%%%%%%%%%%%%%%%%%%%%%%%%%%%%%%%%%%%%%%%%%%%%%%%%%%%%%%%%%%%%%%%%%%%%%%%%%%%%%%%%%%%%%%%%%%%
	
	\begin{quotation}
		\noindent {\it Abstract:}
		  Consider a group of individuals (subjects) participating in the same psychological tests with numerous questions (items) at different times, where the choices of each item have an implicit ordering. The observed responses can be recorded in multiple response matrices over time, named multi-layer ordinal categorical data, where layers refer to time points. Assuming that each subject has a common mixed membership shared across all layers, enabling it to be affiliated with multiple latent classes with varying weights, the objective of the grade of membership (GoM) analysis is to estimate these mixed memberships from the data. When the test is conducted only once, the data becomes traditional single-layer ordinal categorical data. The GoM model is a popular choice for describing single-layer categorical data with a latent mixed membership structure. However, GoM cannot handle multi-layer ordinal categorical data. In this work, we propose a new model, multi-layer GoM, which extends GoM to multi-layer ordinal categorical data. To estimate the common mixed memberships, we propose a new approach, GoM-DSoG, based on a debiased sum of Gram matrices. We establish GoM-DSoG's per-subject convergence rate under the multi-layer GoM model. Our theoretical results suggest that fewer no-responses, more subjects, more items, and more layers are beneficial for GoM analysis. We also propose an approach to select the number of latent classes. Extensive experimental studies verify the theoretical findings and show GoM-DSoG's superiority over its competitors, as well as the accuracy of our method in determining the number of latent classes.

		\vspace{9pt}
		\par
	\end{quotation}\par
\def\thefigure{\arabic{figure}}
\def\thetable{\arabic{table}}

\renewcommand{\theequation}{\thesection.\arabic{equation}}

\fontsize{12}{14pt plus.8pt minus .6pt}\selectfont

%	\begin{frontmatter}
%		\title{An improved spectral clustering method for community detection under the degree-corrected stochastic blockmodel}
%		%\title{A sample article title with some additional note\thanksref{t1}}
%		\runtitle{ISC for Community Detection}
%		%\thankstext{T1}{A sample additional note to the title.}
%		
%		\begin{aug}
%			%%%%%%%%%%%%%%%%%%%%%%%%%%%%%%%%%%%%%%%%%%%%%%
%			%%Only one address is permitted per author. %%
%			%%Only division, organization and e-mail is %%
%			%%included in the address.                  %%
%			%%Additional information can be included in %%
%			%%the Acknowledgments section if necessary. %%
%			%%%%%%%%%%%%%%%%%%%%%%%%%%%%%%%%%%%%%%%%%%%%%%
%			\author[A]{\fnms{Huan} \snm{Qing}\ead[label=e1]{qinghuan@cumt.edu.cn}}
%			\and
%			\author[B]{\fnms{Jingli} \snm{Wang}\ead[label=e2]{jlwang@nankai.edu.cn}}
%			%%%%%%%%%%%%%%%%%%%%%%%%%%%%%%%%%%%%%%%%%%%%%%
%			%% Addresses                                %%
%			%%%%%%%%%%%%%%%%%%%%%%%%%%%%%%%%%%%%%%%%%%%%%%
%			\address[A]{Department of Mathematics,
%				China University of Mining and Technology,
%				\printead{e1}}
%			
%			\address[B]{School of Statistics and Data Science,
%				Nankai University,
%				\printead{e2}}
%		\end{aug}
%

%\end{frontmatter}
%\begin{keyword}[class=MSC]
%\kwd[Primary ]{62H30}
%\kwd{91C20}
%\kwd[; secondary ]{62P25}
%\end{keyword}
%\begin{keyword}
%\kwd{Community detection; spectral clustering; weak signal networks; degree-corrected stochastic blockmodel (DCSBM); regularized Laplacian matrix.}
%%\kwd{\LaTeXe}
%\end{keyword}
%\end{frontmatter}

\section{Introduction}\label{sec1}
Ordinal categorical data is extensively collected in psychological tests, educational assessments, and political surveys \citep{sloane1996introduction,agresti2012categorical,chen2019joint,shang2021partial}.
Such data involves responses that follow a natural and logical order sequence, where the choices can be meaningfully ranked, though the intervals between them may not be precisely measurable or equal. Examples of ordinal categorical data include Likert scale items (strongly disagree/disagree/neutral\\/agree/strongly agree) or disagree/agree in psychological tests, education levels (high school/\\bachelor's degree/master's degree/doctorate) in educational research, and medical assessment scales (mild/moderate/severe) in clinical studies.

Without confusion, we will use categorical data to represent ordinal categorical data throughout the remainder of this paper. This type of data typically comprises responses from subjects to a set of items and is often represented by an $N\times J$ observed response matrix $R$, where $N$ denotes the number of subjects (individuals) participating in the test and $J$ represents the number of items (questions). In this matrix, $R(i,j)$ signifies the response of subject $i$ to item $j$. When each item has only two choices (e.g., yes/no choices or right/wrong answers), the data is referred to as categorical data with binary responses. In this case, $R\in\{0,1\}^{N\times J}$, where $0$ and $1$ are used to indicate the binary choices. When each item has $M$ ($M\geq3$) choices, the data is known as categorical data with polytomous responses. For such data,we consider $R\in\{0,1,\ldots,M\}^{N\times J}$, where integer $m$ denotes the $m$-th ordered choice for $m=0,1,\ldots,M$, with 0 representing the lowest intensity response. In this work, we do not consider missing data. If the original dataset contains missing responses (typically coded as NaN), we recommend removing those subjects with missing values during data preprocessing. It is evident that polytomous responses are more general than binary responses, and thus, this paper focuses on polytomous responses. A key characteristic of traditional categorical data is that it is often collected only once for each subject, resulting in a single-layer categorical dataset. For such data, learning the latent classes (subgroups) of subjects is often of great interest, as subjects with similar subgroup information typically exhibit similar response patterns. Specifically, the task of learning subjects' latent classes is known as latent class analysis (LCA) \citep{hagenaars2002applied,nylund2018ten}, where the latent classes can be different personalities, abilities, and political ideologies in psychological tests, educational assessments, and political surveys \citep{chen2024spectral}, respectively. Therefore, equipped with LCA, categorical data plays a key role in understanding human behaviors, preferences, and attitudes.

However, despite its prevalence, the single-layer structure of traditional categorical data inherently limits the depth of analysis that can be performed, as it captures only a snapshot of the subjects' responses at a single point in time. In many practical scenarios, the same group of subjects may be assessed multiple times over different time points using the same set of items. This repeated assessment generates multi-layer categorical data, with each layer representing a different time point. For such data, all layers share common subjects and items, creating a longitudinal dataset rich in information. The additional temporal dimension offered by multi-layer categorical data provides a more comprehensive view of subjects' responses, potentially revealing patterns and trends that are not discernible in single-layer categorical data. Figure \ref{SimR3} illustrates a simple example of multi-layer categorical data, where subjects' responses exhibit variation across layers.

\begin{figure}
\centering
\resizebox{\columnwidth}{!}{
{\includegraphics[width=0.343\textwidth]{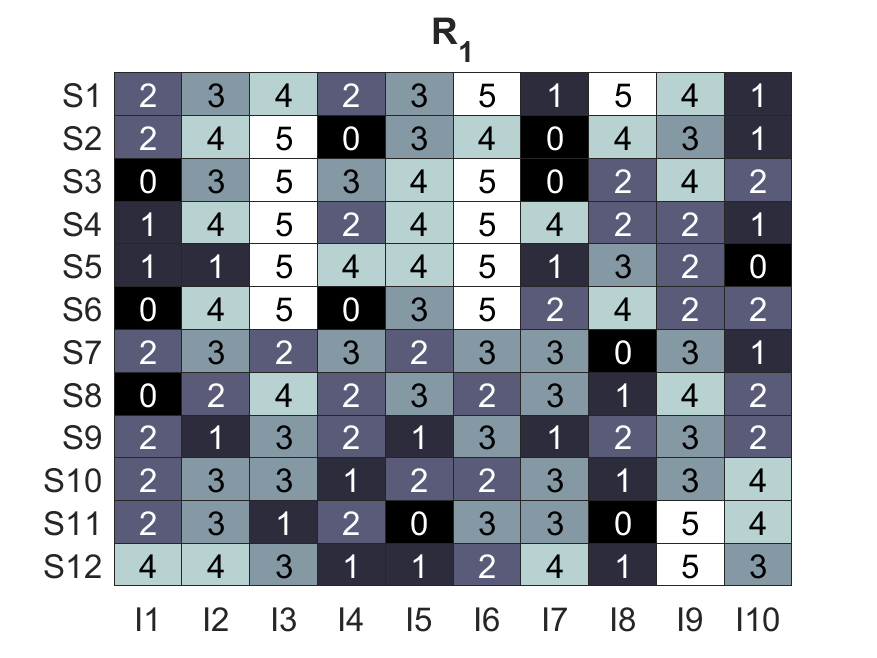}}
{\includegraphics[width=0.343\textwidth]{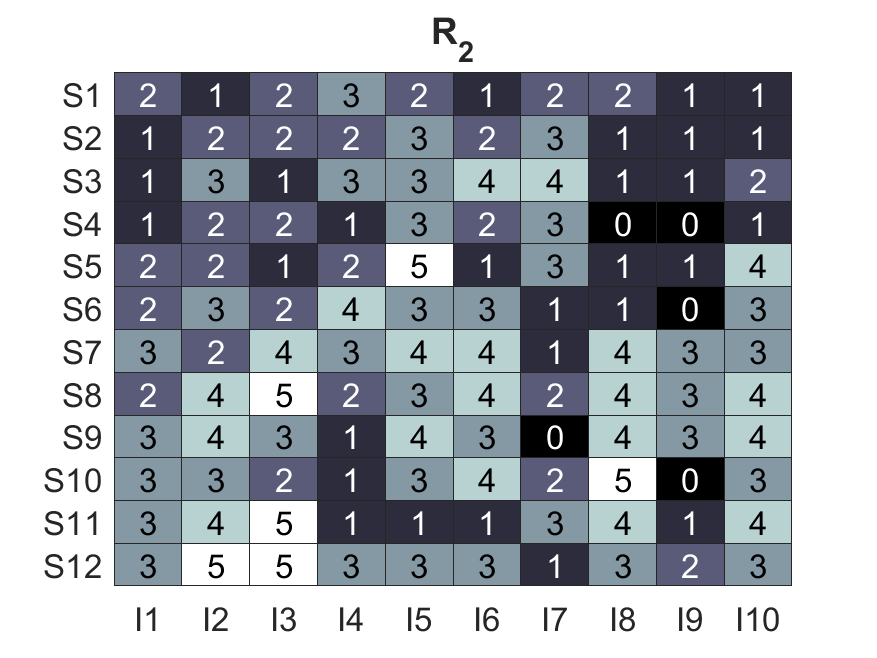}}
{\includegraphics[width=0.343\textwidth]{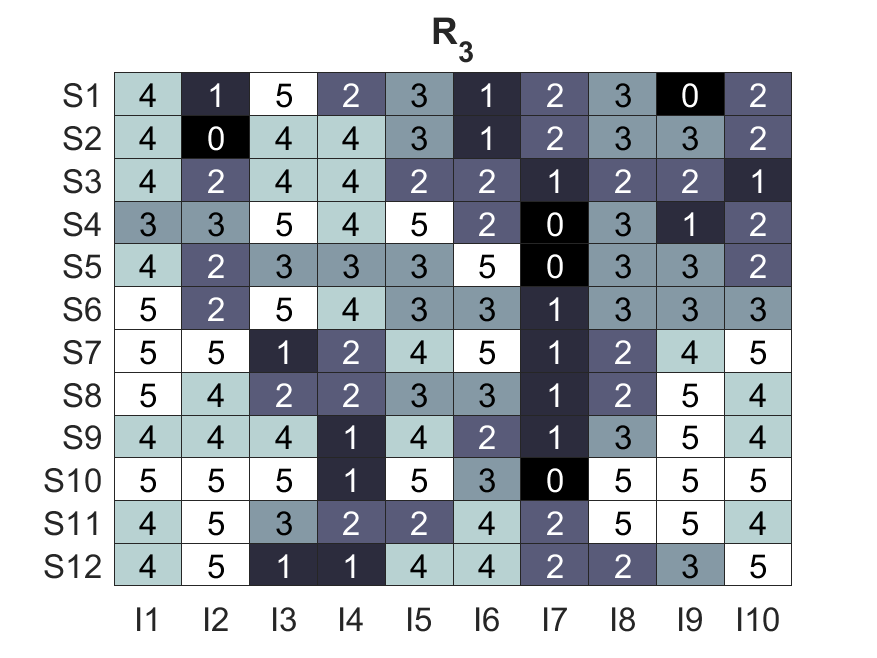}}
}
\caption{A toy example of multi-layer categorical data with 12 subjects, 10 items, 3 layers, and 6 choices per item, where we use S$i$, I$j$, $R_{l}$ to represent the $i$-th subject, the $j$-th item, and the $l$-th observed response matrix, respectively, for $i=1,2,\ldots,12, j=1,2,\ldots,10$, and $l=1,2,3$.}
\label{SimR3} %% label for entire figure
\end{figure}

For single-layer categorical data, a popular statistical model for describing its latent class structure is the latent class model (LCM) \citep{goodman1974exploratory}, which operates under the assumption that each subject is assigned to a single latent class. However, the non-overlapping classes property of LCM can be restrictive, potentially failing to fully capture the complexity and diversity of subjects' response patterns. The grade of membership (GoM) model \citep{woodbury1978mathematical} overcomes LCM's limitation and allows each subject to belong to at least one latent class, providing a more powerful modeling capacity than LCM. Existing approaches for learning subjects' latent mixed memberships in single-layer categorical data modeled by the GoM encompass Bayesian inference utilizing MCMC \citep{erosheva2007describing,gormley2009grade,gu2023dimension}, joint maximum likelihood algorithms \citep{sirt_3.13-194}, and spectral methods \citep{chen2024spectral,qing2024finding,qing2025mixed}.

However, all the above methods are inherently limited to single-layer categorical data and cannot be directly applied to learning the common latent classes of subjects in multi-layer categorical data. The additional temporal dimension in multi-layer categorical data introduces new challenges and opportunities for latent class analysis. Recently, to describe the latent class structures in multi-layer categorical data, \cite{MLLCM} proposed a multi-layer LCM which assumes that each subject belongs to only one latent subgroup. \cite{MLLCM} also provides several spectral methods with theoretical guarantees to discover the latent classes of subjects. However, one main limitation of the multi-layer LCM is that it cannot model multi-layer categorical data with a latent mixed membership structure, where each subject can belong to different subgroups with varying weights.

To address the limitations of the aforementioned models and methods for multi-layer categorical data with a latent mixed membership structure, this work introduces the multi-layer grade of membership (multi-layer GoM) model, which extends the traditional GoM model to multi-layer categorical data. This extension enables us to capture the temporal dynamics of subject responses provided by multiple layers and allows subjects to belong to more than one latent class simultaneously. For the estimation of the common mixed memberships, we develop a novel approach, GoM-DSoG, based on a debiased sum of Gram matrices. This approach effectively handles the challenges posed by multi-layer data. We also establish the per-subject convergence rate of GoM-DSoG under the proposed model. Our theoretical and experimental results highlight the benefits of considering more layers than a single layer in the task of grade of membership analysis for categorical data. In particular, if there are too many zeros in the categorical data, more layers should be considered, and our GoM-DSoG is an efficient tool to estimate subjects' common mixed memberships for such multi-layer categorical data. Additionally, an approach for selecting the number of latent classes $K$ in multi-layer categorical data generated from the proposed model is developed. By conducting substantial experimental studies, we empirically verify the theoretical results and show the superiority of the proposed GoM-DSoG method over its competitors and the high accuracy of our method in choosing $K$.

\emph{Notation:} $I_{m\times m}$ and $[m]$ represent the $m\times m$ identity matrix and set $\{1,2,\ldots,m\}$, respectively. For any matrix $H$, $H'$ denotes its transpose, $H(s,:)$ denotes its submatrix for rows in the set $s$, $\|H\|_{1}$ is its the $l_{1}$ norm, $\|H\|_{F}$ represents its Frobenius norm, $\|H\|$ means its spectral norm, $\|H\|_{\infty}$ represents its maximum absolute row sum, $\|H\|_{2\rightarrow\infty}$ denotes its maximum row-wise $l_{2}$ norm, $\lambda_k(H)$ denotes its $k$-th largest eigenvalue in magnitude, and $\mathrm{rank}(H)$ denotes its rank. $e_{i}$ is a vector satisfying $e_{i}(i)=1$ and $e_{i}(j)=0$ for $j\neq i$. $\mathbb{E}(\cdot)$ denotes expectation and $\mathbb{P}(\cdot)$ represents probability.

We organize the remainder of this paper as follows. We introduce our model in Section \ref{sec2}. Our method for fitting proposed model and its theoretical guarantees are presented in Section \ref{sec3}. Section \ref{sec4} proposes an approach for determining $K$. Section \ref{sec5} conducts experimental studies. Section \ref{sec6} concludes the paper.
\section{Multi-layer grade of membership model}\label{sec2}
\subsection{Data preprocessing}\label{sec:preprocessing}
Before introducing the multi-layer GoM model, we first describe the preprocessing steps applied to the original response matrices to ensure that the responses are appropriately encoded for analysis. Let $R^{\mathrm{original}}_l$ denote the $l$-th layer's original response matrix, which may contain missing values (e.g., NaN). We apply the following preprocessing steps:

\begin{enumerate}
    \item Removal of subjects with missing responses: For each subject $i$, if any entry $R^{\mathrm{original}}_l(i,j)$ is missing (e.g., NaN), we remove the entire subject from the analysis. This ensures that the resulting dataset contains only complete responses.
    \item Recoding of response values: For each non-missing response $R^{\mathrm{original}}_l(i,j)$, we recode the values to ensure they lie in the set $\{0, 1, \ldots, M\}$, where $0$ represents the lowest intensity of response. Specifically:
    \begin{itemize}
        \item If the original responses are in $\{0, 1, \ldots, M\}$, we set $R_l(i,j) = R^{\mathrm{original}}_l(i,j)$. Here, $0$ denotes the lowest intensity response.
        \item If the original responses are in $\{1, 2, \ldots, M+1\}$, we set $R_l(i,j) = R^{\mathrm{original}}_l(i,j) - 1$. This transformation shifts the responses to $\{0, 1, \ldots, M\}$, preserving the ordinal nature of the data without loss of information.
    \end{itemize}
\end{enumerate}

After preprocessing, the resulting observed response matrix $R_l$ for the $l$-th layer has elements in $\{0, 1, \ldots, M\}$, where $0$ represents the lowest response intensity (e.g., ``strongly disagree'' on a Likert scale), and not a missing value. This encoding allows us to model the responses using a Binomial distribution as described in Definition \ref{definMLGoM} given later.
\subsection{The model}
Suppose that all subjects belong to $K$ common latent classes, denoted as:
\begin{align}\label{C1K}
\mathcal{C}_{1}, \mathcal{C}_{2}, \ldots, \mathcal{C}_{K}.
\end{align}

We assume that $K$ is known in this paper and we will also propose an approach for selecting $K$ in Section \ref{sec4}. The memberships of subjects belonging to different latent classes can be characterized by a common membership matrix $\Pi \in [0,1]^{N \times K}$ shared across all layers, where:
\begin{align}\label{Pi}
\Pi(i,k) \text{ is the ``weight'' for the partial membership of subject } i \text{~on~}\mathcal{C}_{k}, \quad i \in [N], k \in [K].
\end{align}

By $\Pi$'s definition, for $i\in[N]$, we have $\sum_{k \in [K]} \Pi(i,k) = 1$. Call subject $i$ pure if one entry of $\Pi(i,:)$ equals 1, and mixed otherwise. In this paper, we assume that:
\begin{align}\label{PureSubject}
\text{There is at least one pure subject in each latent class}.
\end{align}

Since Equations (\ref{C1K})–(\ref{PureSubject}) serve as the foundational assumptions of our model, we adopt them as default conditions throughout this paper. Combining Equations (\ref{Pi}) and (\ref{PureSubject}) implies that $\Pi$ is full rank, and its rank is $K$ since we assume  $K\ll\mathrm{min}(N,J)$ in this paper. Let $p_{k}$ denote the index of any pure subject in the $k$-th latent class $\mathcal{C}_{k}$ for $k \in [K]$. Define the pure subject index set $\mathcal{I}$ as $\mathcal{I} = \{p_{1}, p_{2}, \ldots, p_{K}\}$. Without loss of generality, similar to the analysis in \cite{mao2021estimating}, reorder all subjects such that $\Pi(\mathcal{I},:)=I_{K\times K}$.

For each layer $l \in [L]$, let $\Theta_{l} \in [0, M]^{J \times K}$ denote the item parameter matrix of the $l$-th layer categorical data. Our model, a multi-layer extension of the classical GoM model, is formally defined as follows.
\begin{defin}\label{definMLGoM}
(Multi-layer grade of membership model) Let $R_{l}\in \{0,1,\ldots,M\}^{N\times J}$ represent the observed response matrix of the $l$-th layer categorical data for $l\in[L]$. Our multi-layer grade of membership (multi-layer GoM) model for generating the $L$ observed response matrices $\{R_{l}\}_{l=1}^{L}$ is defined as:
\begin{align}\label{MLGoM}
\mathscr{R}_{l} := \Pi \Theta'_{l}, \quad \text{where } R_{l}(i,j) \sim \text{Binomial}(M, \frac{\mathscr{R}_{l}(i,j)}{M}), \quad i \in [N], j \in [J], l \in [L].
\end{align}
\end{defin}
In our model, $\mathscr{R}_{l}(i,j)$ represents the expected choice of subject $i$ to item $j$ by the property of Binomial distribution, i.e., $\mathbb{E}(R_{l}(i,j))=\mathscr{R}_{l}(i,j)$ under the multi-layer GoM. Equation (\ref{MLGoM}) gives $\mathscr{R}_{l}(i,j)=\sum_{k\in[K]}\Pi(i,k)\Theta_{l}(j,k)$, which means that the response of subject $i$ to item $j$ tends to be larger for larger $\Theta_{l}(j,k)$. For any two distinct latent classes $k_{1}$ and $k_{2}$, if $\Theta_{l}(j,k_{1})>\Theta_{l}(j,k_{2})$, $\mathscr{R}_{l}(i,j)=\sum_{k\in[K]}\Pi(i,k)\Theta_{l}(j,k)$ implies that subjects predominantly affiliated with latent class $k_{1}$ (i.e., those with higher $\Pi(i,k_{1})$) are more likely to choose higher responses for item $j$ at layer $l$, while subjects in latent class $k_{2}$ exhibit weaker or lower responses for the same item. Meanwhile, the choice of modeling categorical responses as Binomial random variables is motivated by the ordered nature of the data. Unlike unordered categorical data (e.g., multiple-choice questions with no inherent order), ordinal categorical data (e.g., Likert scales) can be meaningfully interpreted as counts of ``successes" or levels of agreement. The Binomial distribution provides a flexible framework to model such ordered responses, where the probability of each response value is determined by the latent membership matrix $\Pi$ and item parameter matrix $\Theta_{l}$ for the $l$-th layer. For convenience, call $\mathscr{R}_{l}$ the $l$-th population response matrix since $\mathbb{E}(R_{l})=\mathscr{R}_{l}$ for $l\in[L]$ under the Binomial distribution. When all subjects are pure, our model degenerates the multi-layer LCM introduced in \citep{MLLCM}. When the number of layers $L$ equals 1, the multi-layer GoM defined by Equation (\ref{MLGoM}) simplifies to the GoM model for single-layer categorical data considered in \citep{qing2024finding}. Furthermore, when both $L=1$ and $M=1$, the proposed model reduces to the GoM model studied in \citep{chen2024spectral}.

Equation (\ref{MLGoM}) necessitates that the elements of $\Theta_{l}$ range within $[0,M]$ for $l \in [L]$ because $\frac{\mathscr{R}_{l}(i,j)}{M}$ represents a probability within $[0,1]$, and $\Pi$ satisfies Equation (\ref{Pi}). By the definition of multi-layer GoM, at the $l$-th layer, the probability that subject $i$ selects the $m$-th choice in item $j$ equals
\begin{align}\label{SubProbability}
\mathbb{P}(R_{l}(i,j)=m)=\frac{M!}{m!(M-m)!}(\frac{\mathscr{R}_{l}(i,j)}{M})^{m}(1-\frac{\mathscr{R}_{l}(i,j)}{M})^{M-m},\qquad m=0,1,\ldots,M,
\end{align}
where $\mathbb{P}(R_{l}(i,j)=0)$ represents the probability that subject $i$ tends to have the lowest intensity response to item $j$. By Equation (\ref{MLGoM}), it is clear that the membership matrix $\Pi$ is shared by all layers while the probability that subject $i$ selects the $m$-th choice in item $j$ changes between layers in our model. Such setting is common for the task of community detection in multi-layer networks \citep{paul2020spectral,lei2020consistent,lei2023bias,xu2023covariate,su2024spectral,qing2025community,qing2025discovering}. Using Equation (\ref{MLGoM}), our multi-layer GoM can generate $\{R_{l}\}^{L}_{l=1}$ with a latent true membership matrix $\Pi$ and item parameter matrices $\{\Theta_{l}\}^{L}_{l=1}$. This paper aims to recover $\Pi$ and $\{\Theta_{l}\}^{L}_{l=1}$ given the $L$ observed response matrices $\{R_{l}\}^{L}_{l=1}$.
\subsection{Response intensity scaling}
In real-world multi-layer categorical data with polytomous responses, some lowest-intensity responses (0s) exist. Specifically, for multi-layer categorical data with binary responses when $M=1$, we have $R_{l}\in\{0,1\}^{N\times J}$ for $l\in[L]$, which may result in numerous zeros in the data. A large number of zeros in the dataset poses a significant challenge for the task of grade of membership analysis. Therefore, it is crucial to control the number of zeros in the multi-layer categorical data. Define $\rho=\mathrm{max}_{l\in[L]}\mathrm{max}_{j\in[J], k\in[K]}\Theta_{l}(j,k)$ and $B_{l}=\frac{\Theta_{l}}{\rho}$ for $l\in[L]$. Given that $\Theta_{l}\in[0,M]^{J\times K}$, it follows that $\rho\in(0,M]$, the maximum element of $B_{l}$ is no larger than 1, and $\mathscr{R}_{l}(i,j)\leq\rho$ for $i\in[N], j\in[J], l\in[L]$. According to Equation (\ref{SubProbability}), we have:
\begin{align*}
\mathbb{P}(R_{l}(i,j)=0)=(1-\frac{\mathscr{R}_{l}(i,j)}{M})^{M}=(1-\frac{\rho\Pi(i,:)B'_{l}(j,:)}{M})^{M},
\end{align*}
which indicates that increasing $\rho$ decreases the probability of generating the lowest-intensity responses across all layers. Consequently, $\rho$ serves as a controlling factor for the overall response intensity of the data, thus justifying its designation as the response intensity parameter. To thoroughly examine the performance of the proposed method under varying levels of data's response intensity, we will permit $\rho$ to approach zero and incorporate it into the error bound for further analysis.

\begin{Ex}
Consider a simple case where $M=4$, and the expected response $\mathscr{R}_{l}(i,j)$ takes values from $\{1,1.5,2,2.5,3\}$. Table \ref{ProbabilityR} illustrates how the probability of each response value changes based on $\mathscr{R}_{l}(i,j)$. The results show that as $\mathscr{R}_{l}(i,j)$ increases, the probability of a lowest intensity response (0) decreases monotonically, while the probability of higher intensity response (e.g., 3 or 4) increases. This pattern aligns with theoretical expectations, as $\mathscr{R}_{l}(i,j)$ represents the expected value of $R_{l}(i,j)$ under the proposed model. Intuitively, higher values of $\mathscr{R}_{l}(i,j)$ lead to fewer zeros and shift the response distribution toward higher-valued choices. Consequently, by adjusting the response intensity parameter $\rho$, we can effectively control the overall response intensity of the data.

\begin{table}[h!]
\footnotesize
	\centering
	\caption{Probability distribution of responses under varying population response values $\mathscr{R}_{l}(i,j)$ when $M=4$.}
	\label{ProbabilityR}
\resizebox{\columnwidth}{!}{
	\begin{tabular}{|c|c|c|c|c|c|cccccc}
\hline
Population response values&$\mathbb{P}(R_{l}(i,j)=0)$&$\mathbb{P}(R_{l}(i,j)=1)$&$\mathbb{P}(R_{l}(i,j)=2)$&$\mathbb{P}(R_{l}(i,j)=3)$&$\mathbb{P}(R_{l}(i,j)=4)$\\
\hline
$\mathscr{R}_{l}(i,j)=1$&0.3164&0.4219&0.2109&0.0469&0.0039\\
\hline
$\mathscr{R}_{l}(i,j)=1.5$&0.1526&0.3662&0.3296&0.1318&0.0198\\
\hline
$\mathscr{R}_{l}(i,j)=2$&0.0625&0.2500&0.3750&0.2500&0.0625\\
\hline
$\mathscr{R}_{l}(i,j)=2.5$&0.0198&0.1318&0.3296&0.3662&0.1526\\
\hline
$\mathscr{R}_{l}(i,j)=3$&0.0039&0.0469&0.2109&0.4219&0.3164\\
\hline
\end{tabular}
}
\end{table}
\end{Ex}
\section{Debiased spectral method and its theoretical guarantees}\label{sec3}
We propose an efficient method for $\Pi$ and $\{\Theta_{l}\}^{L}_{l=1}$ in this section. Subsequently, we establish its theoretical guarantees under the multi-layer GoM model.
\subsection{Debiased spectral method}
Suppose that the population response matrices $\{\mathscr{R}_{l}\}^{L}_{l=1}$ are observed, and our objective is to accurately recover $\Pi$ and $\{\Theta_{l}\}^{L}_{l=1}$ from $\{\mathscr{R}_{l}\}^{L}_{l=1}$. Define $\mathcal{S}=\sum_{l=1}^{L}\mathscr{R}_{l}\mathscr{R}'_{l}$, where $\mathcal{S}$ is a symmetric matrix and it represents an aggregation of $L$ Gram matrices $\{\mathscr{R}_{l}\mathscr{R}'_{l}\}^{L}_{l=1}$. The following proposition guarantees the identifiability of our multi-layer GoM model.

\begin{prop}\label{ID}
(Identifiability). Assume that $\mathrm{rank}(\sum_{l=1}^{L}\Theta'_{l}\Theta_{l})=K$. Then, our multi-layer GoM model is identifiable: For any $(\Pi, \{\Theta_{l}\}^{L}_{l=1})$ and $(\tilde{\Pi}, \{\tilde{\Theta}_{l}\}^{L}_{l=1})$, if $\Pi\Theta'_{l}=\tilde{\Pi}\tilde{\Theta}'_{l}$ for $l\in[L]$, we have $\tilde{\Pi}=\Pi\mathcal{P}$ and $\tilde{\Theta}_{l}=\Theta_{l}\mathcal{P}$ for $l\in[L]$, where $\mathcal{P}$ is a permutation matrix.
\end{prop}
The identifiability proposed in Proposition \ref{ID} guarantees that our multi-layer GoM is well-defined since the model parameters $(\Pi, \{\Theta_{l}\}^{L}_{l=1})$ can be uniquely recovered up to a permutation. The intuition of the proposed method comes from the following lemma.
\begin{lem}\label{IS}
When $\mathrm{rank}(\sum_{l=1}^{L}\Theta'_{l}\Theta_{l})=K$, let $\mathcal{S}=U\Lambda U'$ denote the leading $K$ eigendecomposition of $\mathcal{S}$, where $U'U=I_{K\times K}$. Then, we have:
\begin{itemize}
  \item $U=\Pi U(\mathcal{I},:)$,
  \item $\Theta_{l}=\mathscr{R}'_{l}\Pi(\Pi'\Pi)^{-1}$ for $l\in[L]$.
\end{itemize}
\end{lem}
Note that Lemma \ref{IS} is directly linked to the pure subject assumption stated in Equation (\ref{PureSubject}), as it necessitates the presence of the pure subject index set $\mathcal{I}$. Given that the summation of each row of $\Pi$ equals 1 and $\Pi$ satisfies Equation (\ref{PureSubject}), the form $U=\Pi U(\mathcal{I},:)$ is known as the Ideal Simplex (IS), where $U(\mathcal{I},:)$'s $K$ rows are the $K$ vertexes in the simplex structure. Figure \ref{PlotIS} shows the IS structure when $K=3$. It is clear to observe that the rows corresponding to pure subjects form the vertices of the simplex, while those representing mixed subjects lie within its interior. A similar Ideal Simplex is also discovered in the field of mixed membership community detection in network science \citep{mao2021estimating,jin2024mixed,qing2024bipartite}, grade of membership analysis for single-layer categorical data \citep{chen2024spectral,qing2024finding}, and topic modeling \citep{klopp2023assigning,ke2024using}.

\begin{figure}
\centering
%\resizebox{\columnwidth}{!}{
{\includegraphics[width=0.3\textwidth]{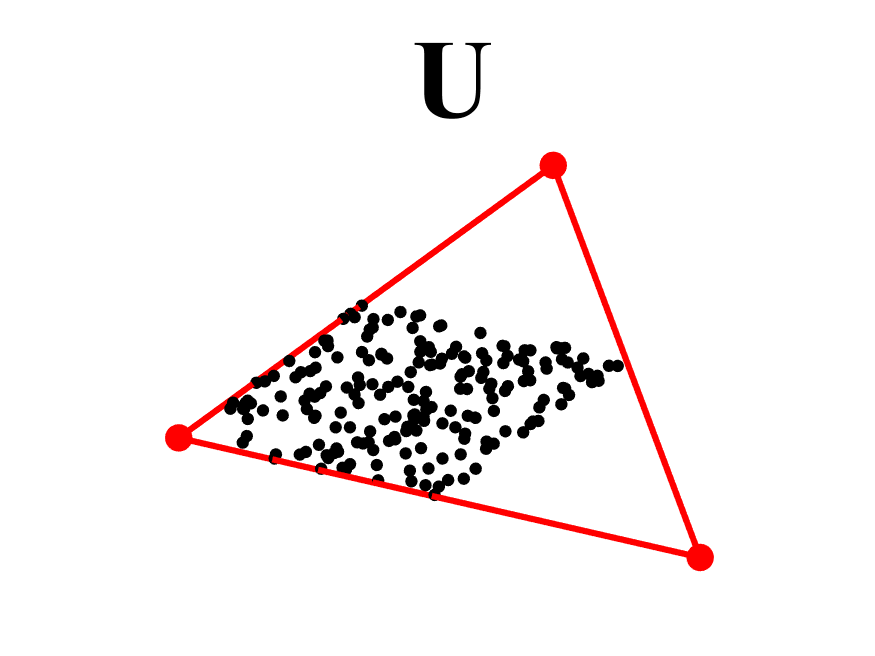}}
%}
\caption{Illustration of the Ideal Simplex geometry of the eigenvector matrix $U$ with $K=3$. Each point denotes a row vector of $U$, red points denote rows corresponding to pure subjects (because $U(i,:)=U(j,:)$ for any two pure subjects $i$ and $j$ from the same latent class, so each red point may denote many pure subjects), each black points corresponds to mixed subjects, and the red triangle denotes the Ideal Simplex. For visualization, these points have been projected from $\mathbb{R}^{3}$ to $\mathbb{R}^{2}$.}
\label{PlotIS} %% label for entire figure
\end{figure}

Since $U(\mathcal{I},:)$ is non-singular, we have $\Pi = U U^{-1}(\mathcal{I},:)$ by $U = \Pi U(\mathcal{I},:)$. This implies that $\Pi$ can be exactly recovered if we can obtain  $U(\mathcal{I},:)$ from the $N \times K$ eigenvector matrix $U$. Similar to previous works related to the IS structure, the \emph{successive projection algorithm} (SPA) \citep{araujo2001successive,gillis2013fast,gillis2015semidefinite} can be employed to precisely identify the $K$ vertices $U(\mathcal{I},:)$ within $U$, leveraging the simplex structure $U = \Pi U(\mathcal{I},:)$. SPA is a widely used technique for vertex hunting in simplex structure \citep{mao2021estimating,jin2024mixed,qing2024bipartite, chen2024spectral,qing2024finding,agterberg2024estimating,klopp2023assigning,ke2024using}. Once we obtain $\Pi$, we can recover $\{\Theta_{l}\}^{L}_{l=1}$ using the second result of Lemma \ref{IS}. To adapt this idea to the real case, we set $H = \mathrm{max}(0, UU^{-1}(\mathcal{I},:)) \equiv\Pi$, ensuring that $\Pi(i,:) = \frac{H(i,:)}{\|H(i,:)\|_{1}}$ and $\Theta_{l} = \mathscr{R}'_{l}\Pi(\Pi'\Pi)^{-1}$ for $i\in[N], l\in [L]$.

For real data, the $L$ response matrices $\{R_{l}\}^{L}_{l=1}$ are observed, while their expectations $\{\mathscr{R}_{l}\}^{L}_{l=1}$ are not available. For each $l \in [L]$, let $D$ be an $N \times N$ diagonal matrix with $D_{l}(i,i)=\sum_{j \in [J]} R_{l}^{2}(i,j)$ for $i \in [N]$. Define $S = \sum_{l \in [L]} (R_{l}R_{l}' - D_{l})$, where $S$ is known as the debiased sum of Gram matrices \citep{su2024spectral}. Let $\hat{U}\hat{\Lambda}\hat{U}'$ be the leading $K$ eigendecomposition of $S$, where $\hat{U}'\hat{U} = I_{K \times K}$. It is known that the $N \times N$ symmetric matrix $S$ is a debiased estimator of $\mathcal{S}$, while $\sum_{l \in [L]} R_{l}R_{l}'$ is a biased estimator, following similar analysis as in \citep{lei2023bias, su2024spectral}. Therefore, let $\hat{\mathcal{I}}$ represent the estimated pure subjects' index set found by employing SPA to $\hat{U}$ with $K$ vertices. Define $\hat{H} = \max(0, \hat{U}\hat{U}^{-1}(\hat{\mathcal{I}},:))$, let $\hat{\Pi}$ be an $N \times K$ matrix such that $\hat{\Pi}(i,:) = \frac{\hat{H}(i,:)}{\|\hat{H}(i,:)\|_{1}}$ for $i \in [N]$. Note that the elements of $R_{l}'\hat{\Pi}(\hat{\Pi}'\hat{\Pi})^{-1}$ may be negative or larger than $M$ while $\Theta_{l}\in[0,M]^{J\times K}$, we define $\hat{\Theta}_{l} = \mathrm{min}(M,\mathrm{max}(0,R_{l}'\hat{\Pi}(\hat{\Pi}'\hat{\Pi})^{-1}))$ for $l \in [L]$. We see that $\hat{\Pi}$ and $\{\hat{\Theta}_{l}\}^{L}_{l=1}$ are good estimations of $\Pi$ and $\{\Theta_{l}\}^{L}_{l=1}$, respectively. The above analysis is summarized in Algorithm \ref{alg:GoM-DSoG}. Recall that the membership scores remain constant across all layers. If we allow the membership scores to vary, Lemma \ref{IS} no longer holds. Consequently, Algorithm \ref{alg:GoM-DSoG} becomes inapplicable for membership score estimation. For alternative approaches to this scenario, readers may consult methodologies developed in \citep{pensky2019spectral,lin2024dynamic}. This paper exclusively addresses the case where $\Pi$ remains consistent across all layers while analyzing the time-varying membership scores is reserved for future research. Additionally, we employ spectral method rather than likelihood-based approaches for estimating membership scores, as spectral method exhibits superior computational efficiency \citep{lei2023bias,chen2024spectral}.

\begin{algorithm}
\caption{\textbf{\textbf{G}}rade \textbf{o}f \textbf{M}embership analysis via the \textbf{D}ebiased \textbf{S}um \textbf{o}f \textbf{G}ram matrices (\textbf{GoM-DSoG})}
\label{alg:GoM-DSoG}
\begin{algorithmic}[1]
\Require The $L$ observed response matrices $R_{1}, R_{2}, \ldots, R_{L}$ and the number of latent classes $K$, where $R_{l}\in\{0,1,2,\ldots, M\}^{N\times J}$ for $l\in[L]$.
\Ensure Estimated membership matrix $\hat{\Pi}$ and estimated item parameter matrices $\{\hat{\Theta}_{l}\}^{L}_{l=1}$.
\State Calculate $S = \sum_{l\in[L]}(R_{l}R_{l}' - D_{l})$, where $D_{l}$ is a diagonal matrix with $D_{l}(i,i) = \sum_{j\in[J]}R_{l}^{2}(i,j)$ for $i\in[N], l\in[L]$.
\State Compute $\hat{U}\hat{\Lambda}\hat{U}'$, the leading $K$ eigendecomposition of $S$.
\State Apply SPA to $\hat{U}$ assuming that there are $K$ vertices to get the estimated pure subjects' index set $\hat{\mathcal{I}}$.
\State Set $\hat{H} = \mathrm{max}(0, \hat{U}\hat{U}^{-1}(\hat{\mathcal{I}},:))$.
\State For $i\in[N]$, set $\hat{\Pi}(i,:) = \frac{\hat{H}(i,:)}{\|\hat{H}(i,:)\|_{1}}$.
\State For $l\in[L]$, set $\hat{\Theta}_{l} = \mathrm{min}(M,\mathrm{max}(0,R_{l}'\hat{\Pi}(\hat{\Pi}'\hat{\Pi})^{-1}))$.
\end{algorithmic}
\end{algorithm}
\begin{rem}
In this paper, our multi-layer GoM model employs a Binomial distribution in Equation (\ref{MLGoM}) for analyzing multi-layer ordinal categorical data. Notably, extending the community detection approaches for weighted networks proposed in \citep{qing2023community,qing2024bipartite}, our framework allows $\mathcal{F}$ to represent any distribution whose first-moment exists and satisfies $\mathbb{E}(R_{l}(i,j))=\mathscr{R}_{l}(i,j)$ with $\mathscr{R}_{l}:=\Pi\Theta'_{l}$ for $i\in[N], j\in[J], l\in[L]$.This generalization enables the modeling of multi-layer weighted data. Below we present illustrative examples:
\begin{itemize}
  \item \textbf{Normal distribution}: Let $\mathcal{F}$ be a Normal distribution such that $R_{l}(i,j)\sim\mathrm{Normal}(\mathscr{R}_{l}(i,j),\sigma^{2})$, where $\sigma^{2}$ denotes variance. For this distribution, $R_{l}(i,j)$ can be any real value, $\Theta_{l}$'s elements can be negative, and $\rho\in(0,+\infty)$.
  \item \textbf{Poisson distribution}: Let $\mathcal{F}$ be a Poisson distribution such that $R_{l}(i,j)\sim\mathrm{Poisson}(\mathscr{R}_{l}(i,j))$. Here, $R_{l}(i,j)$ is a nonnegative integer, $\Theta_{l}$'s elements should be nonnegative, and $\rho\in(0,+\infty)$.
  \item \textbf{Uniform distribution}: Let $\mathcal{F}$ be a Uniform distribution such that $R_{l}(i,j)\sim\mathrm{Uniform}(0,2\mathscr{R}_{l}(i,j))$. For this distribution, $R_{l}(i,j)$ is  positive value that is no larger than $2\rho$, $\Theta_{l}$'s elements should be nonnegative, and $\rho\in(0,+\infty)$.
  \item \textbf{Discrete signed distribution}: Let $\mathcal{F}$ be a discrete distribution such that $\mathbb{P}(R_{l}(i,j)=1)=\frac{1+\mathscr{R}_{l}(i,j)}{2}$ and $\mathbb{P}(R_{l}(i,j)=-1)=\frac{1-\mathscr{R}_{l}(i,j)}{2}$. For this case, $R_{l}(i,j)$ takes values in $\{-1,1\}$ (i.e., signed responses), $\Theta_{l}$'s elements should range in $[-1,1]$, and $\rho\in(0,1]$.
\end{itemize}

$\mathcal{F}$ can also be any other distribution as long as $\mathbb{E}(R_{l})=\mathscr{R}_{l}$ holds under distribution $\mathcal{F}$. Crucially, Lemma \ref{IS} remains valid under these generalized distributions, enabling consistent estimation of $\Pi$ and $\{\Theta_{l}\}^{L}_{l=1}$ using our GoM-DSoG algorithm. When different distributions $\mathcal{F}$ are considered, our multi-layer GoM can model multi-layer weighed categorical data because $R_{l}$'s elements have more choices now. While this extension broadens applicability to multi-layer weighted data through diverse response value domains, our current focus remains on multi-layer ordinal categorical data using the Binomial distribution. The theoretical investigation of GoM-DSoG's performance under different distributions $\mathcal{F}$ is of independent interest.

\end{rem}
\subsection{Consistency results}
This subsection presents the theoretical findings exploring the advantages of multi-layer categorical data over single-layer categorical data for grade of membership analysis. We investigate the asymptotic consistency of GoM-DSoG under the proposed multi-layer GoM model across different intensity levels. Specifically, we assume that $N$, $J$, and $L$ can approach infinity in the asymptotic setting, with no assumed relationship among their growth rates. Before presenting the main result, we outline some assumptions.
\begin{assum}\label{Assum1}
$\rho^{2}NJL\gg\mathrm{log}(N+J+L)$.
\end{assum}
Since $N$, $J$, and $L$ can grow to infinity in our analysis, Assumption \ref{Assum1} is mild, as it implies that $\rho$ should be significantly larger than $\sqrt{\frac{\log(N+J+L)}{NJL}}$ (therefore, $\rho$ can be sufficiently small for large values of $N$, $J$, and $L$). This assumption imposes a lower bound on $\rho$, indicating that there cannot be too many zeros in the multi-layer categorical data. It is worth noting that when $J=O(N)$, our Assumption \ref{Assum1} is in agreement with the sparsity requirement stated in Theorem 1 of \cite{lei2023bias}. This alignment is expected since both our GoM-DSoG and the Bias-adjusted SoS algorithm provided in \cite{lei2023bias} are debiased spectral clustering methods.
\begin{assum}\label{Assum11} $|\lambda_{K}(\sum_{l\in[L]}B'_{l}B_{l})|\geq c_{1}JL$ for some positive constant $c_{1}$.
\end{assum}
This assumption requires a linear growth of $|\lambda_{K}(\sum_{l\in[L]}B'_{l}B_{l})|$ concerning the number of layers $L$ and it serves a similar purpose in our theoretical analysis as the second point of Assumption 1 in \cite{lei2023bias}, Assumption 2 in \cite{su2024spectral}, and the last point of Condition 2 in \cite{chen2024spectral}. The rationality of this assumption can be illustrated by considering a simple case where $B_{1}=B_{2}=\ldots=B_{L}$. Under this scenario, it follows that $|\lambda_{K}(\sum_{l\in[L]}B'_{l}B_{l})|=L|\lambda_{K}(B'_{1}B_{1})|=O(JL)$, which demonstrates the validity of the assumption in a straightforward manner.
\begin{con}\label{condition}
$K=O(1)$ and $\lambda_{K}(\Pi'\Pi)=O(\frac{N}{K})$.
\end{con}
This condition serves a similar purpose to the first point of Assumption 1 in \cite{lei2023bias}, Condition 2 in \cite{chen2024spectral}, Assumption 1 in \cite{su2024spectral}, and so on. It is necessary to simplify our theoretical analysis.

Theorem \ref{MAIN} presents our main result for the task of grade membership analysis. It provides the per-node error rate of the proposed algorithm. To obtain Theorem \ref{MAIN}, we first bound $\|S-\mathcal{S}\|_{\infty}$ using the Bernstein inequality provided in \cite{tropp2012user}. Subsequently, we utilize this bound to establish a bound for the row-wise eigenvector error $\|\hat{U}\hat{U}'-UU'\|_{2\rightarrow\infty}$ based on Theorem 4.2 of \cite{cape2019the}.
\begin{thm}\label{MAIN}
Consider a multi-layer GoM parameterized by $\{\Pi,\{\Theta_{l}\}^{L}_{l=1}\}$. When Assumption \ref{Assum1}, Assumption \ref{Assum11}, and Condition \ref{condition} hold, with high probability, we have
\begin{align*}
\max_{i\in[N]}\|e'_{i}(\hat{\Pi}-\Pi\mathcal{P})\|_{1}=O(\sqrt{\frac{\log(N+J+L)}{\rho^{2}NJL}})+O(\frac{1}{N}).
\end{align*}
\end{thm}
The result presented in Theorem \ref{MAIN} indicates that enhancing the values of $\rho$, $N$, and $J$ enhances the performance of GoM-DSoG in estimating $\Pi$, which aligns with our intuition. In particular, a higher intensity parameter $\rho$ reduces the number of low-intensity answers, providing more informative and varied data. This makes the distinct response patterns of each latent class more pronounced and easier for the algorithm to distinguish, leading to more accurate parameter estimates. Moreover, it indicates that as $L$ increases, GoM-DSoG demonstrates improved performance for the task of GoM analysis. This observation underscores the benefits of considering multi-layer categorical data over single-layer categorical data, as the former contains more responses and, consequently, more information. Through the proof of Theorem \ref{MAIN}, it is evident that the impact of the three model parameters $\rho, N, J$, and $L$ cannot be analyzed without Assumption \ref{Assum1}, Assumption \ref{Assum11}, and Condition \ref{condition}. This underscores that these assumptions are employed for theoretical purposes. Furthermore, Theorem \ref{MAIN} relies critically on the pure subject assumption in Equation (\ref{PureSubject}), as this ensures the existence of the Ideal Simplex structure in Lemma \ref{IS}. When Lemma \ref{IS} holds, applying SPA to $U$ can exactly identify the $K$ true pure subjects, implying that GoM-DSoG will exactly recover $\Pi$ and $\{\Theta_{l}\}^{L}_{l=1}$ when we use $\mathcal{S}$ to replace $S$ in GoM-DSoG. When Equation (\ref{PureSubject}) holds, we can apply the SPA algorithm to $\hat{U}$ to estimate pure subjects when $\{R_{l}\}^{L}_{l=1}$ are generated from the proposed model with parameters $\Pi$ and $\{\Theta_{l}\}^{L}_{l=1}$. Finally, we can bound the difference between $\Pi$ and $\hat{\Pi}$. Otherwise, if Equation (\ref{PureSubject}) does not hold (say, there is at least one latent class having no pure subjects), GoM-DSoG can not exactly recover $\Pi$ and $\{\Theta_{l}\}^{L}_{l=1}$ even when we use $\mathcal{S}$ to replace $S$ in GoM-DSoG, which leads to a result that we can not obtain Theorem \ref{MAIN} anymore.

For the task of estimating the item parameter matrices $\{\Theta_{l}\}^{L}_{l=1}$, the following lemma provides a bound on $\frac{\|\sum_{l\in[L]}(\hat{\Theta}_{l}-\Theta_{l}\mathcal{P})\|_{F}}{\|\sum_{l\in[L]}\Theta_{l}\|_{F}}$.
\begin{lem}\label{BoundItemParameterMatrices}
Consider a multi-layer GoM parameterized by $\{\Pi,\{\Theta_{l}\}^{L}_{l=1}\}$. Under the conditions that $\rho L\mathrm{max}(N,J)\gg\mathrm{log}(N+J+L)$ and $|\lambda_{1}(\sum_{l\in[L]}B_{l})|\geq c_{2}\sqrt{J}L$ for some constant $c_{2}>0$, with high probability, we have
\begin{align*}
\frac{\|\sum_{l\in[L]}(\hat{\Theta}_{l}-\Theta_{l}\mathcal{P})\|_{F}}{\|\sum_{l\in[L]}\Theta_{l}\|_{F}}=O(\sqrt{\frac{\mathrm{max}(N,J)\mathrm{log}(N+J+L)}{\rho NJL}}).
\end{align*}
\end{lem}
Note that the conditions specified in Lemma \ref{BoundItemParameterMatrices} differ from those of Theorem \ref{MAIN}. This difference arises because bounding $\|\sum_{l\in[L]}(R_{l}-\mathscr{R}_{l})\|$ is necessary to obtain the result presented in Lemma \ref{BoundItemParameterMatrices}. Despite the distinct conditions, the result in Lemma \ref{BoundItemParameterMatrices} also indicates that enhancing the values of $\rho, N, J$, and $L$ leads to improved performance of the GoM-DSoG method.
\section{Estimation of $K$}\label{sec4}
Selecting the appropriate number of latent classes $K$ in multi-layer categorical data with a latent mixed membership structure poses a significant challenge. While Algorithm \ref{alg:GoM-DSoG} assumes $K$ is known, this is often not feasible in practical applications. Given $L$ observed responses $\{R_{l}\}^{L}_{l=1}$, our GoM-DSoG algorithm can always produce an estimated mixed membership matrix $\hat{\Pi}_{k}\in[0,1]^{N\times k}$ for any number of latent classes $k$. When faced with multiple potential choices for $K$, such that the true $K$ may be one of $\{1,2,\ldots, K_{c}\}$, we can run our GoM-DSoG algorithm on the data with $k$ latent classes for all $k\in[K_{c}]$, resulting in $K_{c}$ estimated mixed membership matrices $\{\hat{\Pi}_{k}\}^{K_{c}}_{k=1}$. Naturally, this leads to the question of which estimated membership matrix indicates the best quality of mixed membership partition.

To address this question, we propose a metric called averaged fuzzy modularity, inspired by the fuzzy modularity introduced in \citep{qing2024finding} for single-layer categorical data with a latent mixed membership structure. The averaged fuzzy modularity based on the estimated mixed membership matrix $\hat{\Pi}_{k}$ when we assume there are $k$ latent classes is defined as follows:
\begin{align}\label{averagedfuzzymodularity}
Q(\hat{\Pi}_{k})=\frac{1}{L}\sum_{l\in[L]}\sum_{i\in[N]}\sum_{j\in[N]}\frac{1}{\eta_{l}}(G_{l}(i,j)-\frac{g_{l}(i,i)g_{l}(j,j)}{\eta_{l}})\hat{\Pi}_{k}(i,:)\hat{\Pi}'_{k}(j,:),
\end{align}
where $G_{l}=R_{l}R'_{l}$, $g_{l}$ is an $N\times N$ diagonal matrix with $g_{l}(i,i)=\sum_{j\in[N]}G_{l}(i,j)$, and $\eta_{l}=\sum_{j\in[N]}g_{l}(j,j)$ for $i\in[N], l\in[L]$. This metric has several notable properties:
\begin{itemize}
  \item When there are no mixed subjects, our modularity reduces to the averaged modularity of \citep{paul2021null, MLLCM}.
  \item When $L=1$, it reduces to the fuzzy modularity of \citep{nepusz2008fuzzy, qing2024finding}.
  \item When $L=1$ and there are no mixed subjects, it becomes the popular Newman-Girvan modularity \citep{newman2004finding}.
\end{itemize}

A larger value of $Q$ indicates better quality of mixed membership partition \citep{newman2006modularity,nepusz2008fuzzy,paul2021null}. Therefore, similar to \citep{nepusz2008fuzzy}, we infer the optimal $K$ by choosing the one that maximizes $Q$ defined in Equation (\ref{averagedfuzzymodularity}).
\section{Experimental studies}\label{sec5}
\subsection{Simulations}
This part investigates the performance of the proposed GoM-DSoG method for the task of the grade of membership analysis and its accuracy in the estimation of $K$. To this end, we conduct five experiments to study the effects of various factors, including the number of subjects $N$, the number of layers $L$, the response intensity parameter $\rho$, the number of pure subjects in each latent class, and the number of latent classes $K$. In all experiments, the GoM-DSoG method is compared with the following two different algorithms:
\begin{itemize}
  \item \textbf{GoM-Sum}: Set $R_{\mathrm{sum}}=\sum_{l\in[L]}R_{l}$. Let $\hat{U}_{R_{\mathrm{sum}}}\hat{\Sigma}_{R_{\mathrm{sum}}}\hat{V}'_{R_{\mathrm{sum}}}$ represent $R_{\mathrm{sum}}$'s leading $K$ singular value decomposition, where $\hat{U}'_{R_{\mathrm{sum}}}\hat{U}_{R_{\mathrm{sum}}}=I_{K\times K}$. This approach estimates the parameters $\Pi$ and $\{\Theta_{l}\}^{N}_{l=1}$ by substituting $R_{\mathrm{sum}}$ for $S$ and $\hat{U}_{R_{\mathrm{sum}}}$ for $\hat{U}$ in Algorithm \ref{alg:GoM-DSoG}.
  \item \textbf{GoM-SoG}: Set $\tilde{S}=\sum_{l\in[L]}R_{l}R'_{l}$. This method estimates $\Pi$ and $\{\Theta_{l}\}^{N}_{l=1}$ by replacing $S$ with $\tilde{S}$ in Algorithm \ref{alg:GoM-DSoG}.
\end{itemize}
\begin{rem}
Define $\mathscr{R}_{\mathrm{sum}} = \sum_{l \in [L]} \mathscr{R}_{l}$. When the rank of the $J\times K$ aggregation matrix $\sum_{l \in [L]}B_{l}$ is $K$, let $U_{\mathscr{R}_{\mathrm{sum}}} \Sigma_{\mathscr{R}_{\mathrm{sum}}} V'_{\mathscr{R}_{\mathrm{sum}}}$ denote the leading $K$ singular value decomposition of $\mathscr{R}_{\mathrm{sum}}$ such that $U'_{\mathscr{R}_{\mathrm{sum}}} U_{\mathscr{R}_{\mathrm{sum}}} = I_{K \times K}$. The intuition behind GoM-Sum stems from the observation that $U_{\mathscr{R}_{\mathrm{sum}}} = \Pi U_{\mathscr{R}_{\mathrm{sum}}}(\mathcal{I},:)$ also exhibits a simplex structure. The intuition behind GoM-SoG is analogous to that of GoM-DSoG, as $\tilde{S}$ serves as a biased estimator of $\mathcal{S}$ \citep{lei2023bias, su2024spectral, MLLCM}. By following a similar analysis as in \citep{MLLCM}, we can obtain theoretical guarantees for GoM-Sum and GoM-SoG. Furthermore, we can prove that our proposed GoM-DSoG consistently outperforms both GoM-Sum and GoM-SoG. For brevity, we omit the detailed theoretical analysis of GoM-Sum and GoM-SoG in this paper.
\end{rem}

To assess the performance of all approaches in estimating $\Pi$, the \textbf{Relative $l_{1}$ error} defined as $\mathrm{min}_{\mathcal{P}\in\mathcal{X}}\frac{\|\hat{\Pi}-\Pi\mathcal{P}\|_{1}}{N}$ is used, where $\mathcal{X}$ collects all $K\times K$ permutation matrices. To measure their performances in estimating $\{\Theta_{l}\}^{L}_{l=1}$, we use the \textbf{Relative $l_{2}$ error}, defined as $\mathrm{min}_{\mathcal{P}\in\mathcal{X}}\frac{\|\sum_{l\in[L]}(\hat{\Theta}_{l}-\Theta_{l}\mathcal{P})\|_{F}}{\|\sum_{l\in[L]}\Theta_{l}\|_{F}}$. To evaluate the accuracy of all methods in estimating $K$, the \textbf{Accuracy rate} calculated by the proportion of times an approach accurately chooses the true $K$ by maximizing $Q$ is employed. The Relative $l_{1}$ error and Relative $l_{2}$ error metrics are such that smaller values are better, while the Accuracy rate metric is such that larger values are better.

For all experiments, we set $J=\frac{N}{5}$ and $M=5$, so that $R_{l}\in\{0,1,2,3,4,5\}^{N\times \frac{N}{5}}$ for all $l\in[L]$. Let each latent class have $N_{0}$ pure subjects. For a mixed subject $i$, set $\Pi(i,k)=\frac{\mathrm{rand}(1)}{K-1}$ for $k\in[K-1]$ and $\Pi(i,K)=1-\sum_{k\in[K-1]}\Pi(i,k)$ when $K\geq2$, where $\mathrm{rand}(1)$ denotes a random value from $(0,1)$. For $j\in[J], k\in[K], l\in[L]$, set $B_{l}(j,k)=\mathrm{rand}(1)$. For all experiments, $N, K, \rho, N_{0}$, and $L$ are set independently. For each parameter setting, we report the averaged metric over the 50 repetitions. For all approaches, $K$ is given for estimating $\Pi$ and $\{\Theta_{l}\}^{L}_{l=1}$ while $K$ is not given and is estimated by maximizing the averaged fuzzy modularity value $Q$ for the task of determining $K$. Meanwhile, for each experiment, there are two cases: the low-intensity case and the high-intensity case. In the low-intensity case, we let the response intensity parameter $\rho$ be small, while in the high-intensity case, $\rho$ is large. The  low-intensity case is characterized by numerous zeros in the generated data, which aligns multi-layer categorical data with binary responses or the scenario where many participants have lowest responses to all items in each test. Conversely, the high-intensity case features only a few zeros, corresponding to real-world situations where only a minority of participants have lowest responses in multi-layer categorical data with polytomous responses.

\textbf{Experiment 1: Effect of changing $N$.} Set $L=5$, $N_{0}=N/5$, and $K=3$. Let $N$ vary in $\{100, 200, \ldots, 1000\}$. For the low-intensity case, set $\rho=0.2$. For the high-intensity case, set $\rho=5$. Results are presented in Figure \ref{EX1}, indicating that (a) all methods enjoy better performance as the number of subjects $N$ grows; (b) our GoM-DSoG slightly outperforms GoM-SoG, and both methods significantly outperform GoM-Sum, especially in the low-intensity case; (c) all methods perform better in the high-intensity case compared to the low-intensity case for estimating $\Pi$, $\{\Theta_{l}\}^{L}_{l=1}$, and $K$. Notably, all methods demonstrate satisfactory performance in the high-intensity case, which is consistently observed in other experiments (thus, similar analyses are omitted for subsequent experiments).
\begin{figure}
\centering
%\resizebox{\columnwidth}{!}{
\subfigure[$\rho=0.2$]{\includegraphics[width=0.3\textwidth]{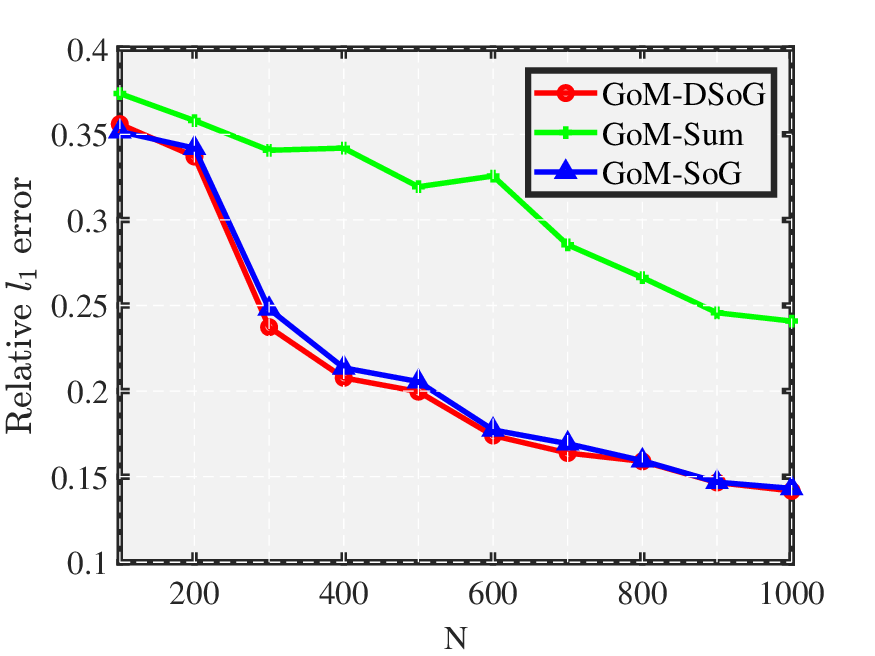}}
\subfigure[$\rho=0.2$]{\includegraphics[width=0.3\textwidth]{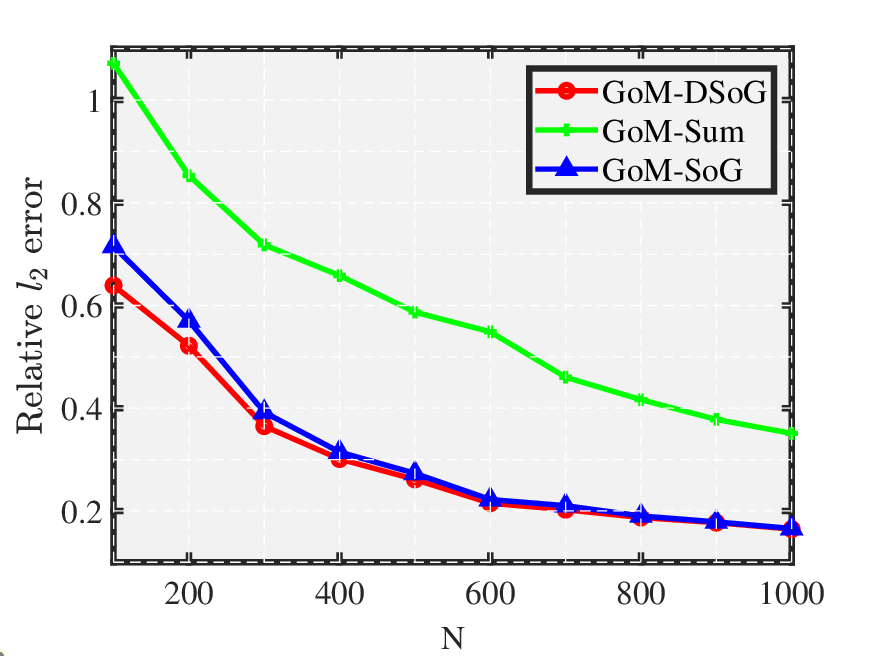}}
\subfigure[$\rho=0.2$]{\includegraphics[width=0.3\textwidth]{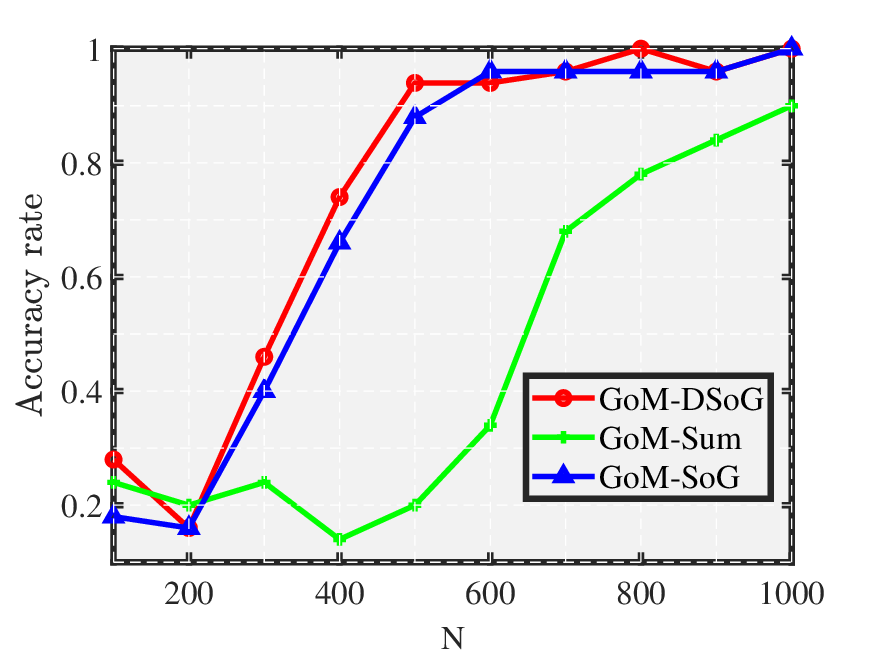}}
\subfigure[$\rho=5$]{\includegraphics[width=0.3\textwidth]{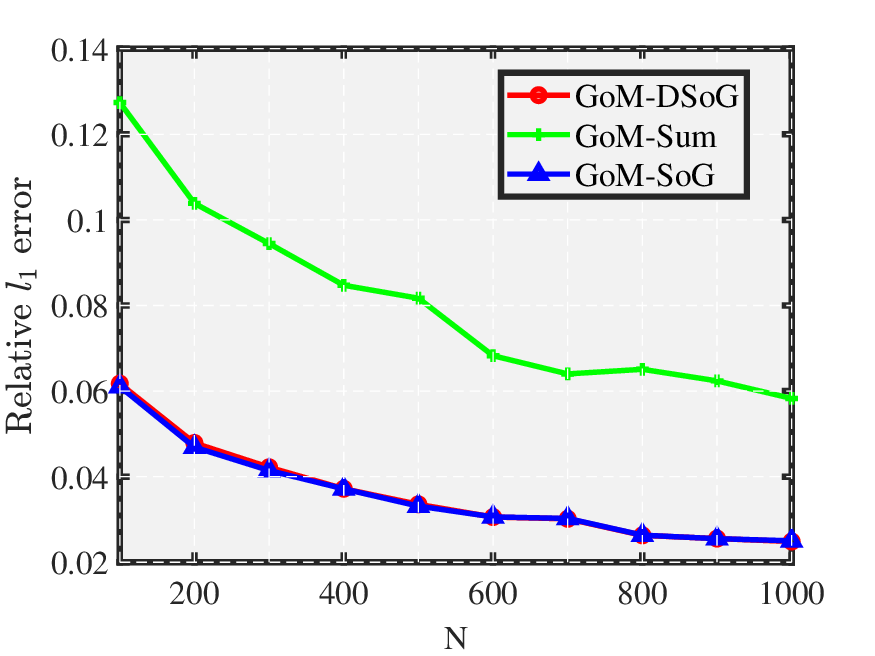}}
\subfigure[$\rho=5$]{\includegraphics[width=0.3\textwidth]{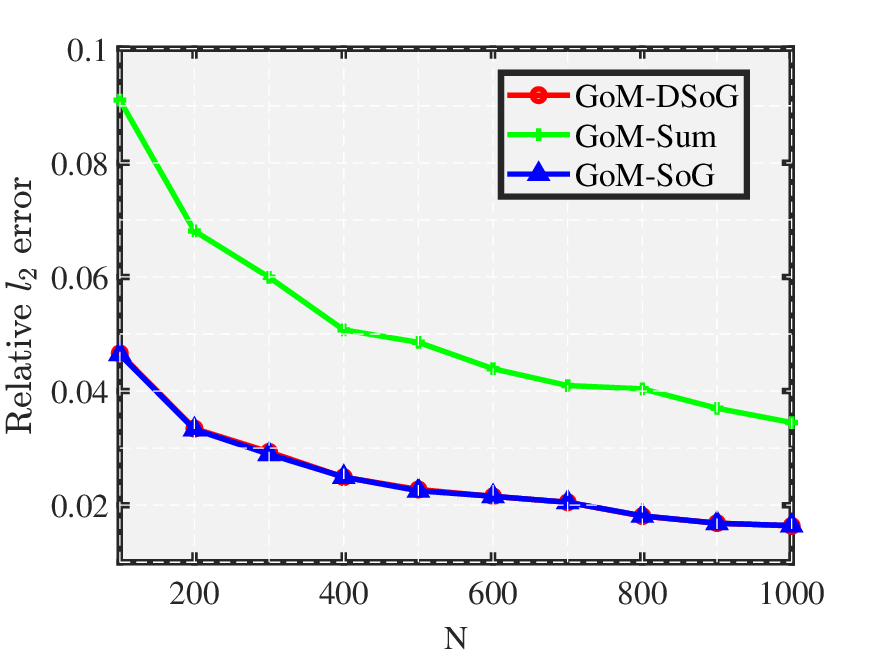}}
\subfigure[$\rho=5$]{\includegraphics[width=0.3\textwidth]{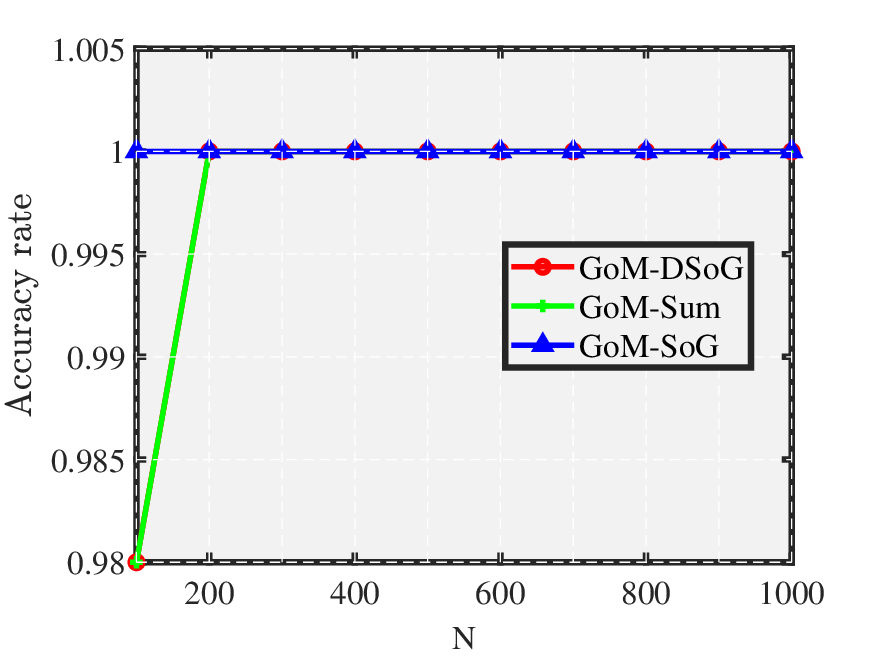}}
%}
\caption{Experiment 1.}
\label{EX1} %% label for entire figure
\end{figure}

\textbf{Experiment 2: Effect of changing $L$.} Set $N=500$, $N_{0}=100$, and $K=3$. Let $L$ vary in $\{1, 2, \ldots, 10\}$. For the low-intensity  case, set $\rho=0.2$. For the high-intensity case, set $\rho=5$. Results are reported in Figure \ref{EX2}. We observe that GoM-DSoG and GoM-SoG perform similarly, both methods behave better as $L$ increases, and they consistently outperform GoM-Sum. The results of this experiment support the advantages of considering multiple tests (i.e., layers) at different times. Specifically, if a significant proportion of lowest responses is present in the data, conducting additional tests at various times, i.e., increasing $L$, can be beneficial. This process results in multi-layer categorical data. However, for such data, simply summing all response matrices (i.e., using the GoM-Sum method) for the grade of membership analysis and estimating $K$ is inefficient. Instead, our proposed GoM-DSoG method is more effective.
\begin{figure}
\centering
%\resizebox{\columnwidth}{!}{
\subfigure[$\rho=0.2$]{\includegraphics[width=0.3\textwidth]{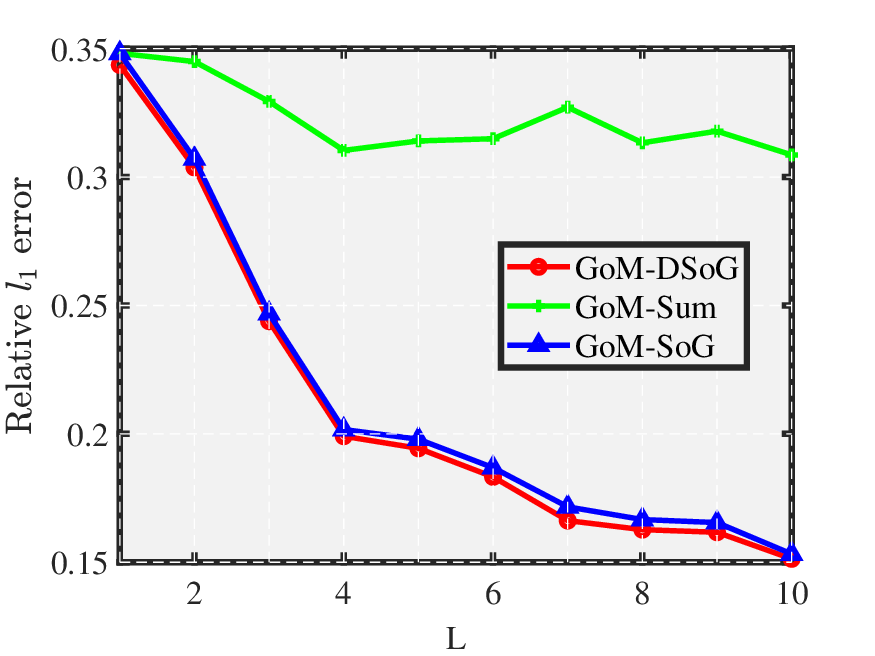}}
\subfigure[$\rho=0.2$]{\includegraphics[width=0.3\textwidth]{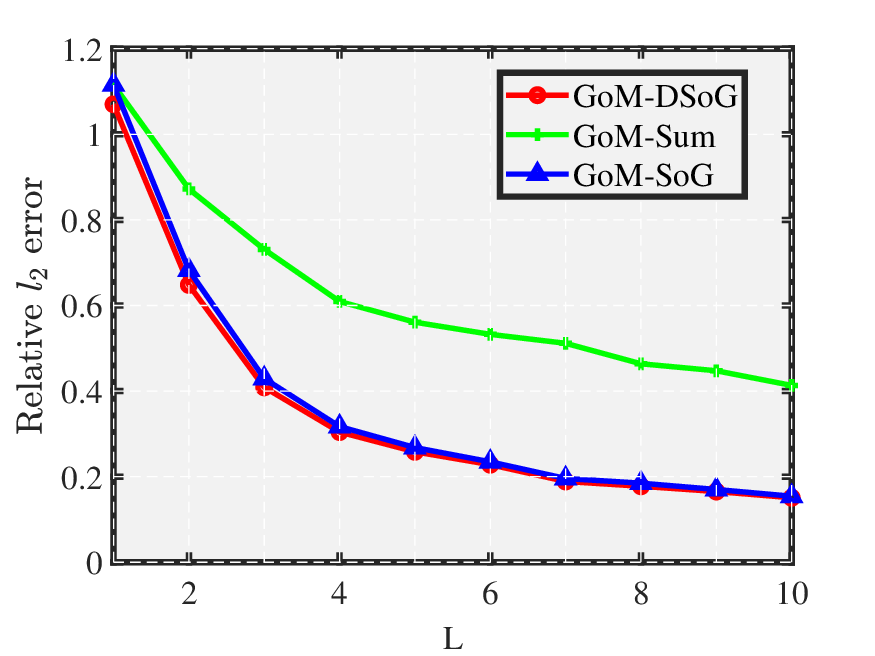}}
\subfigure[$\rho=0.2$]{\includegraphics[width=0.3\textwidth]{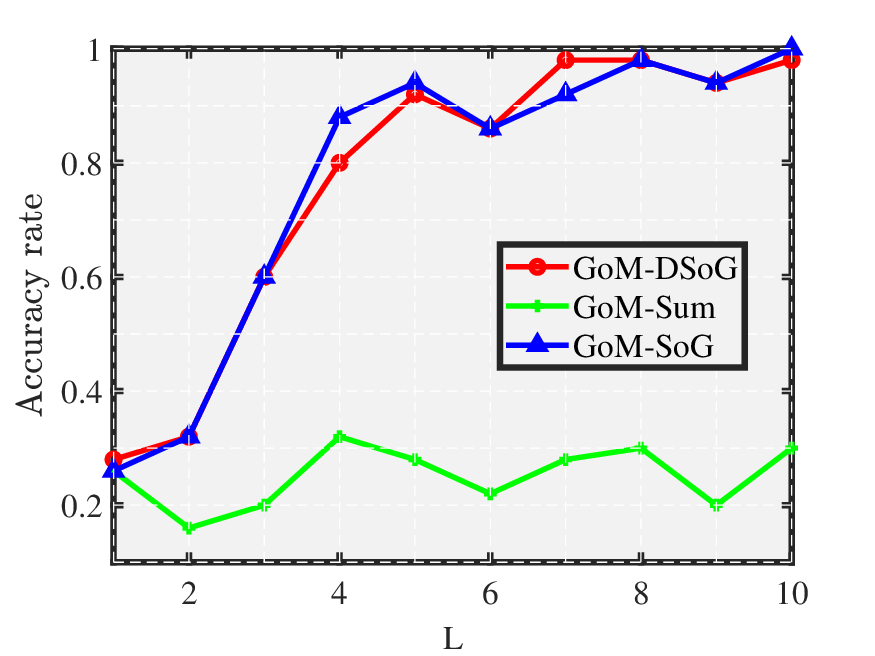}}
\subfigure[$\rho=5$]{\includegraphics[width=0.3\textwidth]{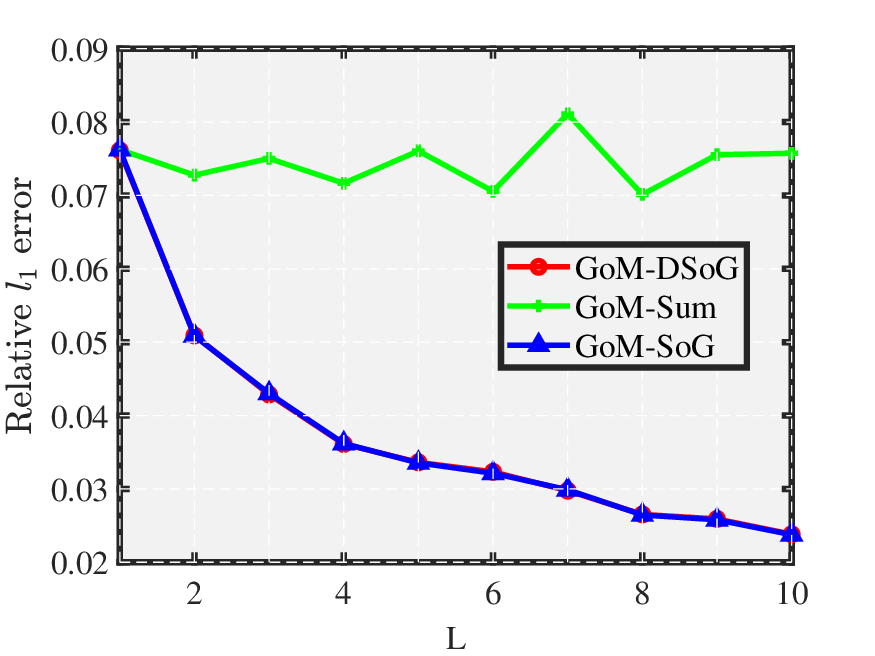}}
\subfigure[$\rho=5$]{\includegraphics[width=0.3\textwidth]{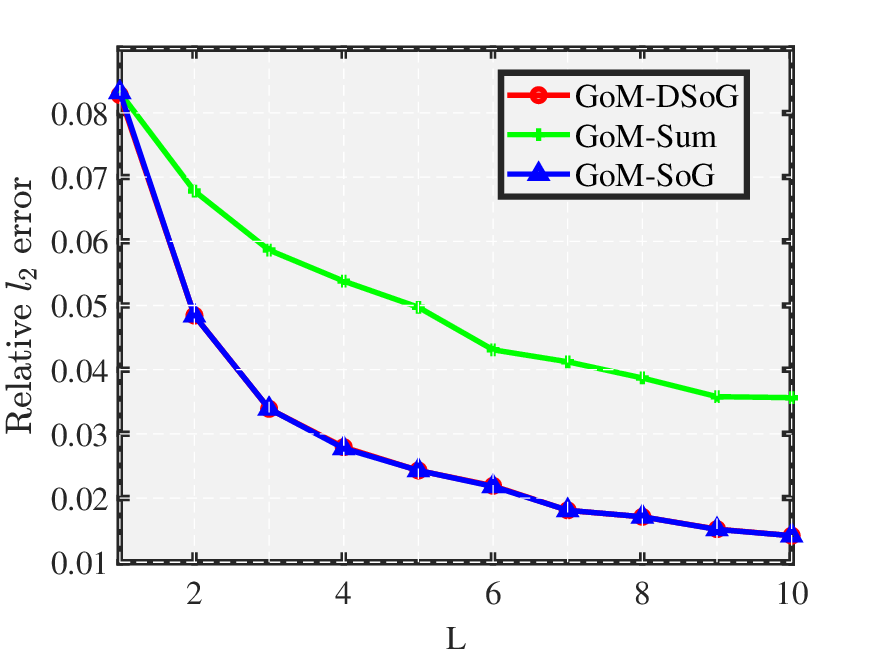}}
\subfigure[$\rho=5$]{\includegraphics[width=0.3\textwidth]{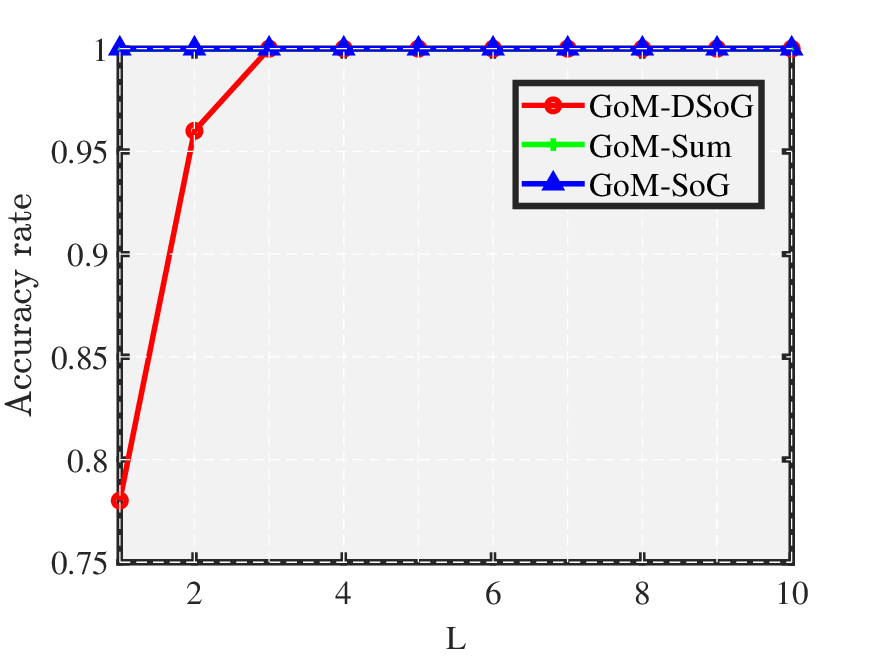}}
%}
\caption{Experiment 2.}
\label{EX2} %% label for entire figure
\end{figure}

\textbf{Experiment 3: Effect of changing $\rho$.} Set $N=500$, $L=5$, $N_{0}=100$, and $K=3$. For the low-intensity  case, let $\rho$ vary in $\{0.05, 0.1, \ldots, 0.5\}$. For the high-intensity case, let $\rho$ vary in $\{0.5, 1, \ldots, 5\}$. Results are displayed in Figure \ref{EX3}. Observations indicate that all methods perform better as the response intensity parameter $\rho$ increases, our GoM-DSoG slightly outperforms GoM-SoG, and both significantly outperform GoM-Sum. We also plot the estimated simplex structure returned by applying the SPA algorithm to the eigenvector matrix $\hat{U}$ under varying response intensity parameters $\rho$ in Figure \ref{EX3IS}. As $\rho$ increases, the points representing pure subjects (denoted by red points) become increasingly concentrated near the vertices of the simplex structure, while the points representing mixed subjects (denoted by black points) tend to distribute more uniformly within the interior of the simplex. This observation aligns with the theoretical expectation that higher values of $\rho$ reduce the number of zero entries in the data, thereby improving the separability of pure and mixed subjects and making it easier for SPA to accurately detect pure subjects. We further report the Accuracy rate of SPA in identifying pure subjects in Figure \ref{EX3SPArho}, where the Accuracy rate is defined as the proportion of times, over 50 repetitions, that applying SPA to $\hat{U}$ correctly identifies $K$ pure subjects, each from a distinct latent class. From Figure \ref{EX3SPArho}, it is evident that the Accuracy rate of SPA increases as $\rho$ increases. When $\rho$ is small, the data has a significant number of zero entries. This intensity introduces noise and makes it challenging for SPA to distinguish pure subjects from mixed subjects. However, as $\rho$ increases, the data has more higher responses, and the inherent structure of the simplex becomes more pronounced. This allows SPA to more accurately identify pure subjects. These results highlight SPA’s vertex stability—its ability to consistently and reliably in hunting the vertices of the simplex structure, even under varying levels of data intensity controlled by the parameter $\rho$.
\begin{figure}
\centering
%\resizebox{\columnwidth}{!}{
\subfigure[$\rho\in\{0.05,0.1,\ldots,0.5\}$]{\includegraphics[width=0.3\textwidth]{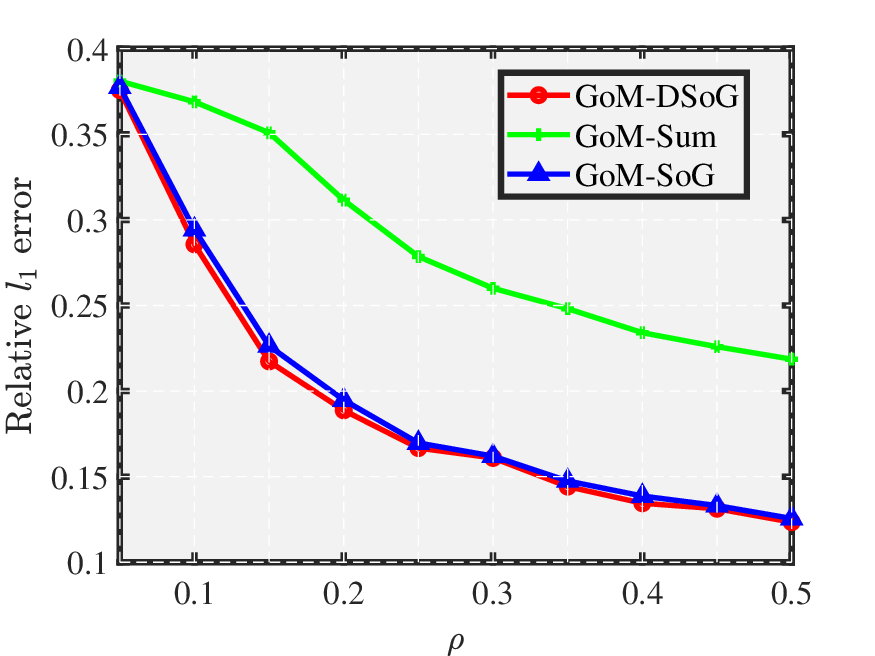}}
\subfigure[$\rho\in\{0.05,0.1,\ldots,0.5\}$]{\includegraphics[width=0.3\textwidth]{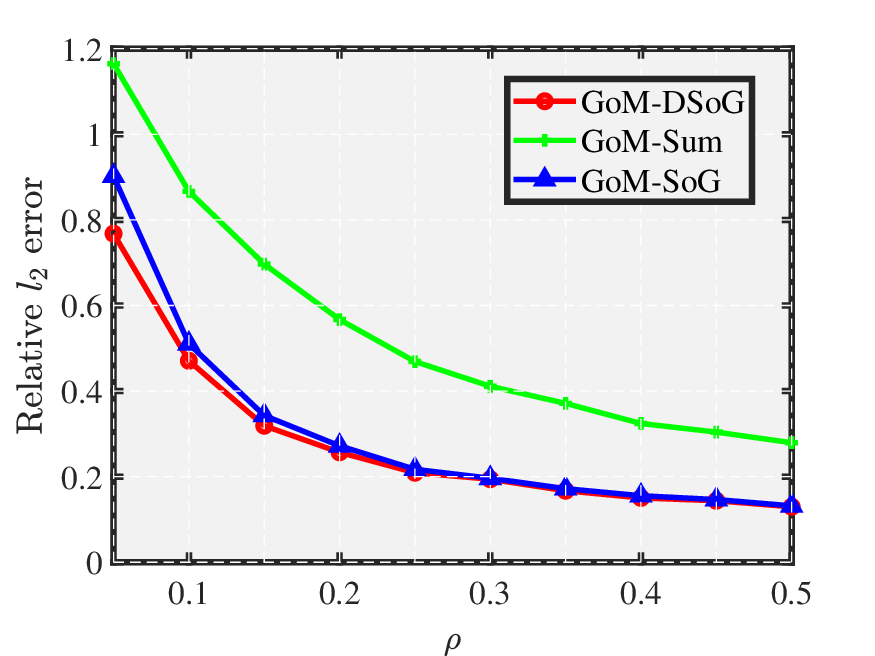}}
\subfigure[$\rho\in\{0.05,0.1,\ldots,0.5\}$]{\includegraphics[width=0.3\textwidth]{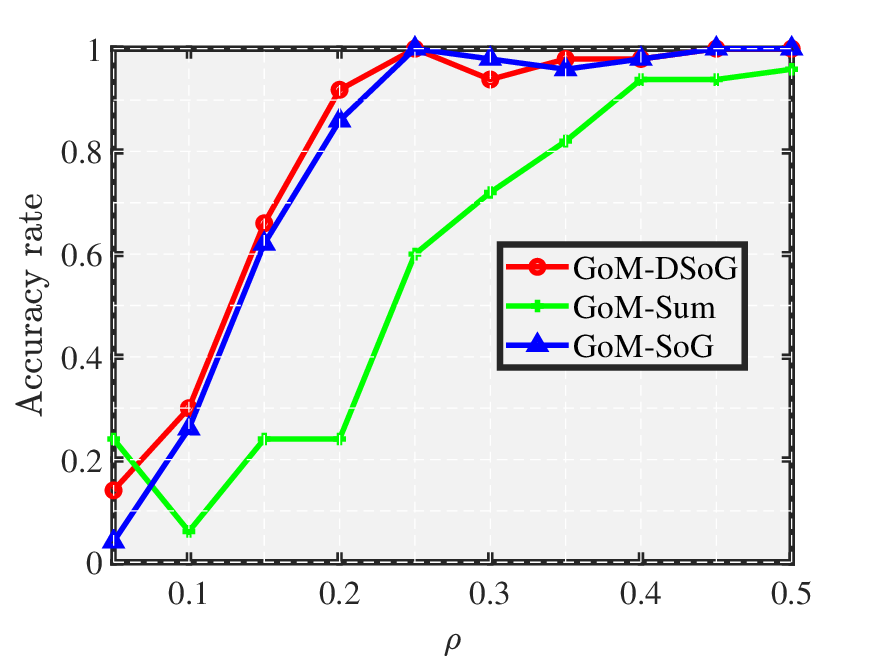}}
\subfigure[$\rho\in\{0.5,1,\ldots,5\}$]{\includegraphics[width=0.3\textwidth]{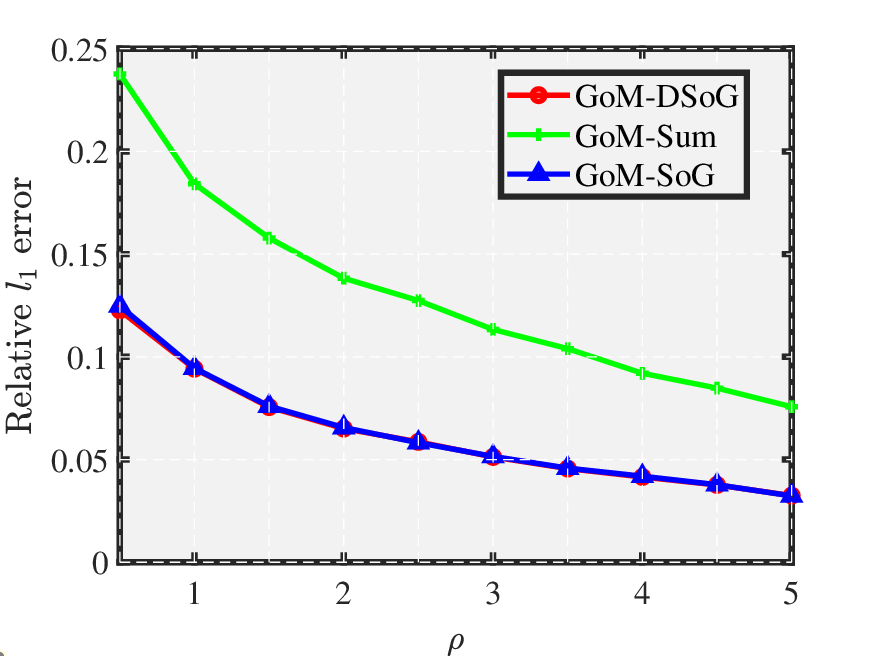}}
\subfigure[$\rho\in\{0.5,1,\ldots,5\}$]{\includegraphics[width=0.3\textwidth]{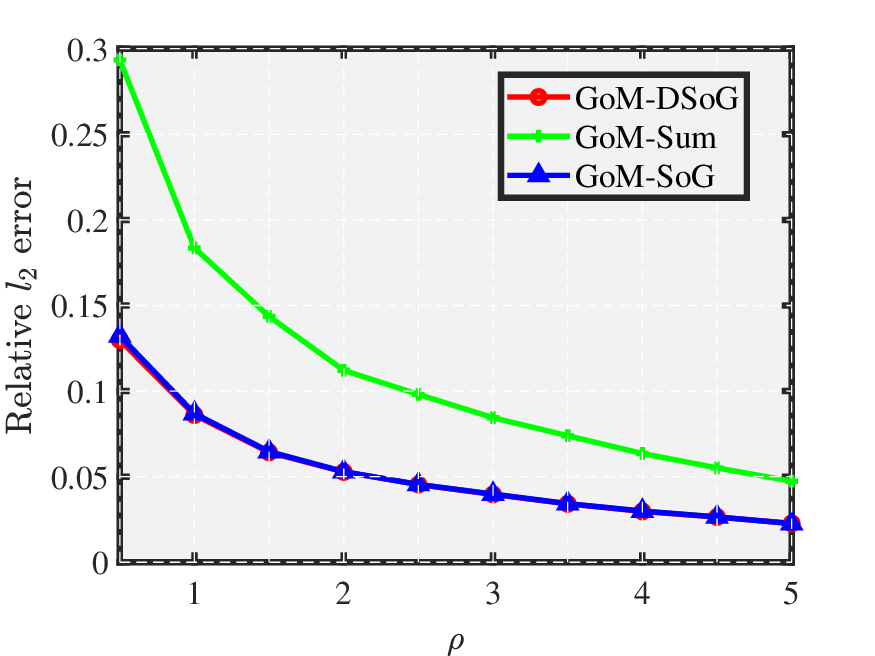}}
\subfigure[$\rho\in\{0.5,1,\ldots,5\}$]{\includegraphics[width=0.3\textwidth]{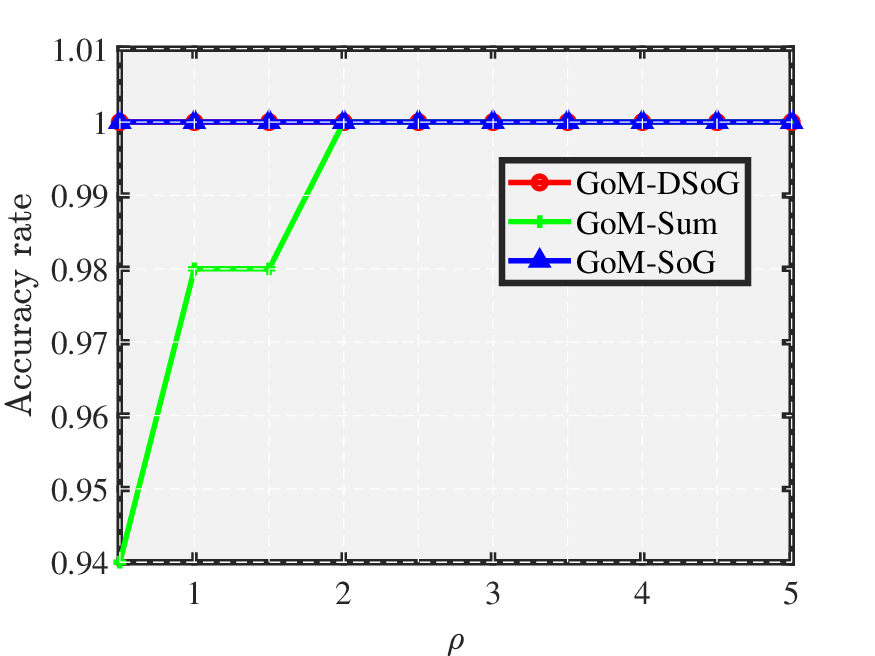}}
%}
\caption{Experiment 3.}
\label{EX3} %% label for entire figure
\end{figure}

\begin{figure}
\centering
%\resizebox{\columnwidth}{!}{
{\includegraphics[width=0.32\textwidth]{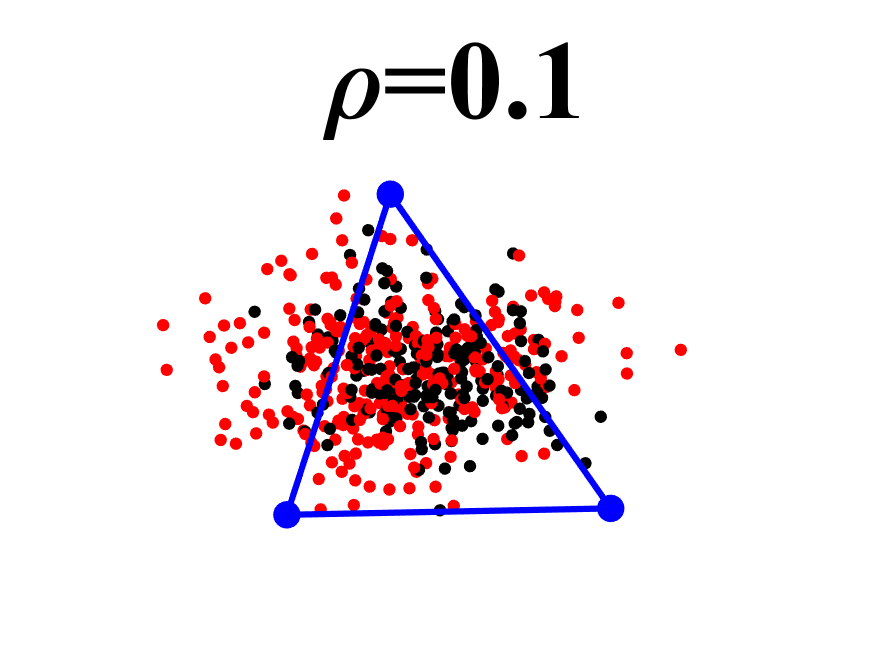}}
{\includegraphics[width=0.32\textwidth]{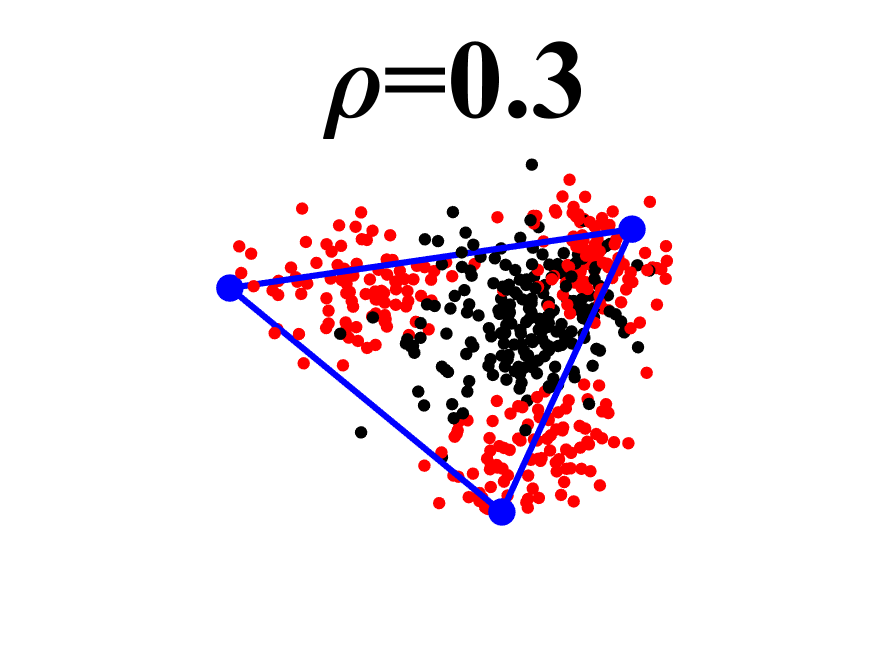}}
{\includegraphics[width=0.32\textwidth]{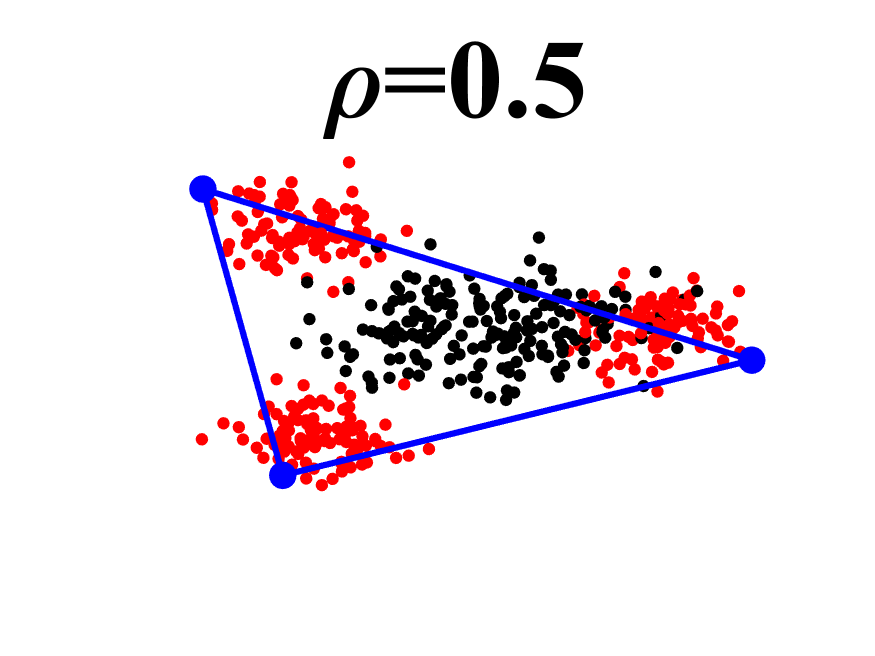}}
{\includegraphics[width=0.32\textwidth]{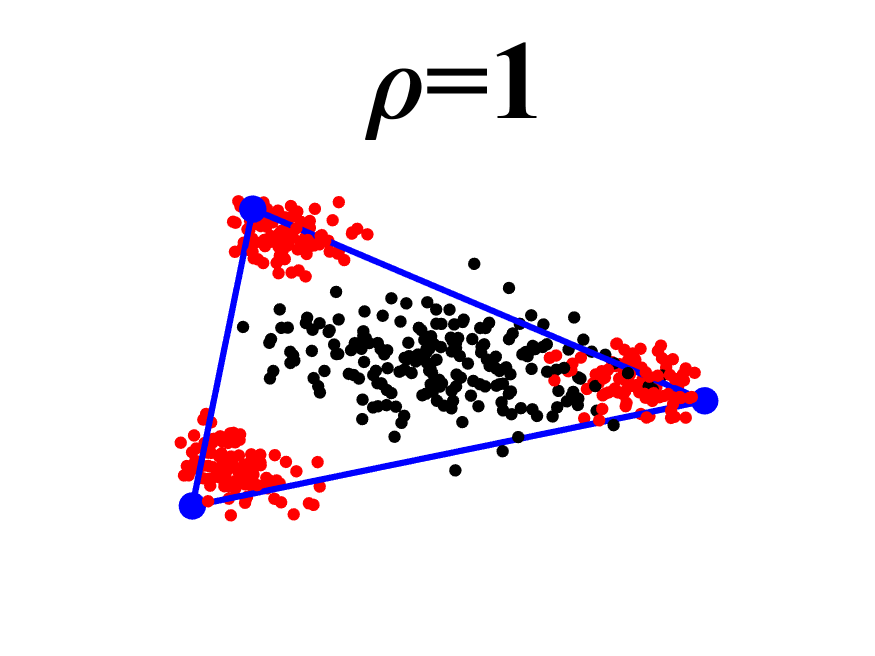}}
{\includegraphics[width=0.32\textwidth]{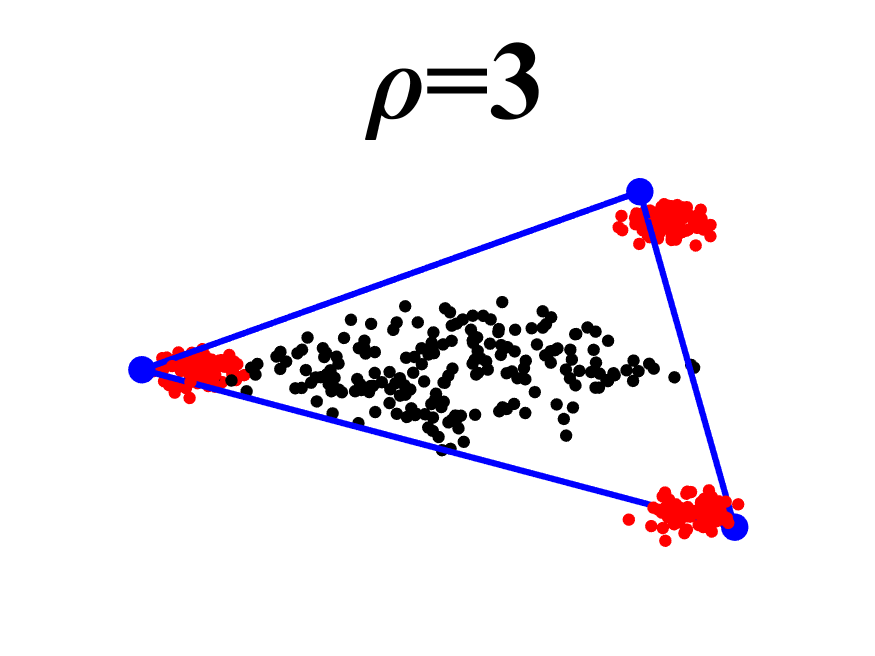}}
{\includegraphics[width=0.32\textwidth]{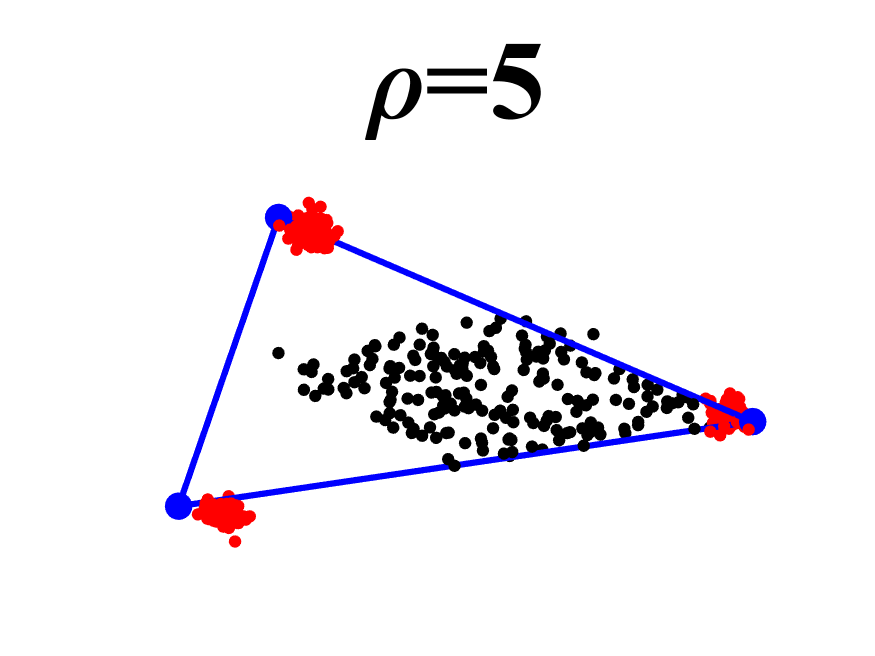}}
%}
\caption{Illustration of the estimated simplex structure for different values of the response intensity parameter $\rho$. In each sub-figure, each point represents a row of $\hat{U}$ obtained from the debiased sum of Gram matrices $S$. Red points represent pure subjects, while black points represent mixed subjects. The three blue points denote the estimated pure subjects obtained by applying the SPA algorithm to $\hat{U}$, and the blue triangle represents the estimated simplex structure. For visualization, these points have been projected from $\mathbb{R}^{3}$ to $\mathbb{R}^{2}$.}
\label{EX3IS} %% label for entire figure
\end{figure}

\begin{figure}
\centering
%\resizebox{\columnwidth}{!}{
{\includegraphics[width=0.4\textwidth]{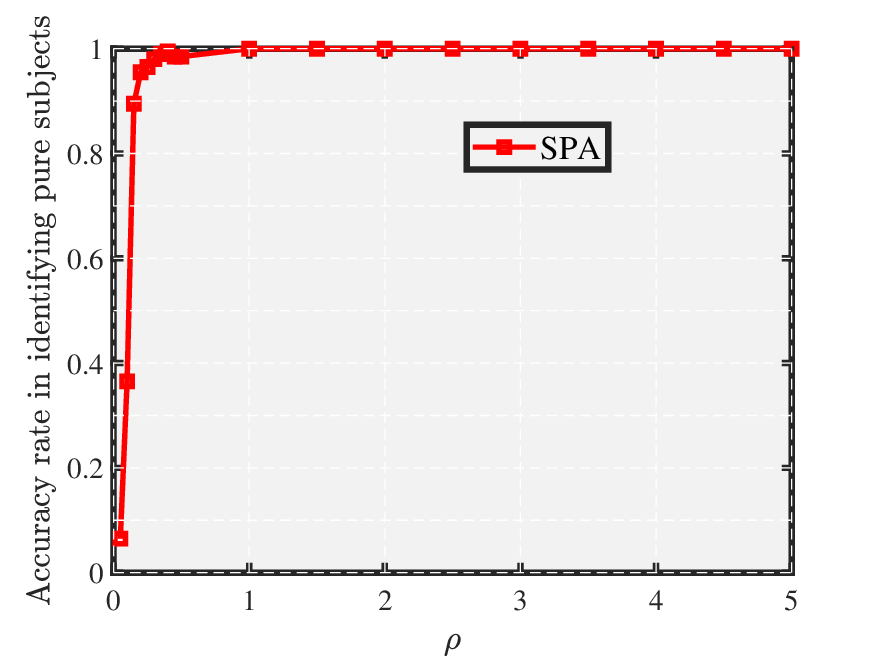}}
%}
\caption{Accuracy rate of the SPA algorithm in identifying pure subjects when $\rho$ takes values in $\{0.05, 0.1, \ldots, 0.5, 1,1.5,2,\ldots,5\}$.}
\label{EX3SPArho} %% label for entire figure
\end{figure}

\textbf{Experiment 4: Effect of changing $N_{0}$.} Set $N=600$, $L=10$, and $K=3$. Let $N_{0}$ vary in $\{20, 40, \ldots, 200\}$. For the low-intensity case, set $\rho=0.2$. For the high-intensity case, set $\rho=5$. Performances of the three methods are reported in Figure \ref{EX4}. Results suggest that all methods exhibit improved performance in estimating the item parameter matrices and $K$ when $N_{0}$ grows, with GoM-DSoG performing slightly better than GoM-SoG and both significantly outperforming GoM-Sum.
\begin{figure}
\centering
%\resizebox{\columnwidth}{!}{
\subfigure[$\rho=0.2$]{\includegraphics[width=0.3\textwidth]{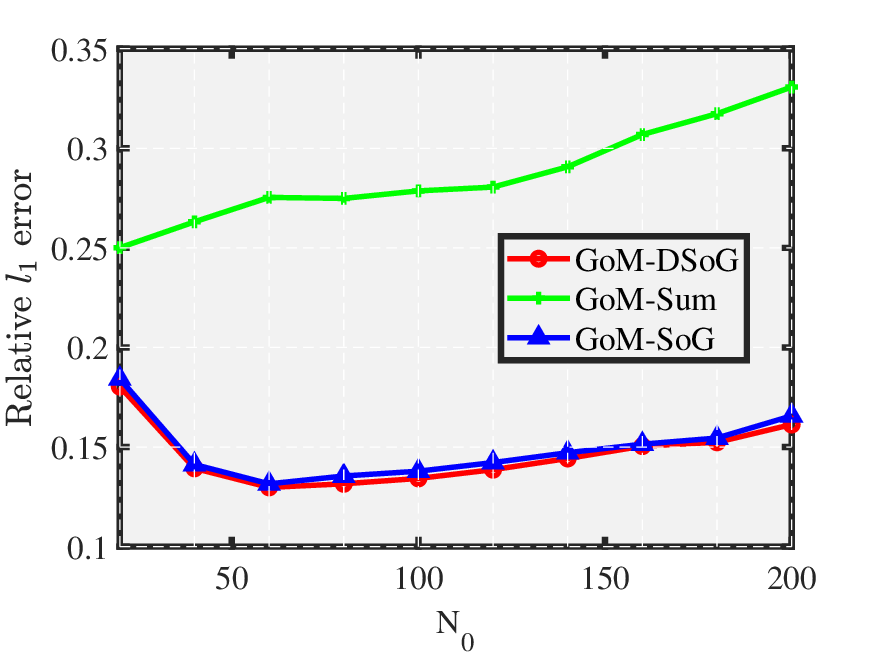}}
\subfigure[$\rho=0.2$]{\includegraphics[width=0.3\textwidth]{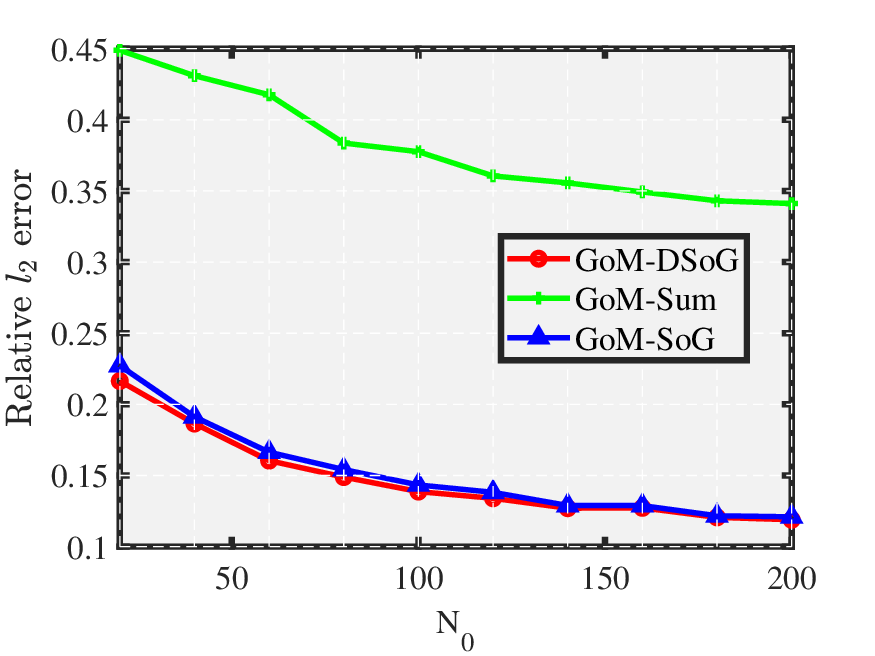}}
\subfigure[$\rho=0.2$]{\includegraphics[width=0.3\textwidth]{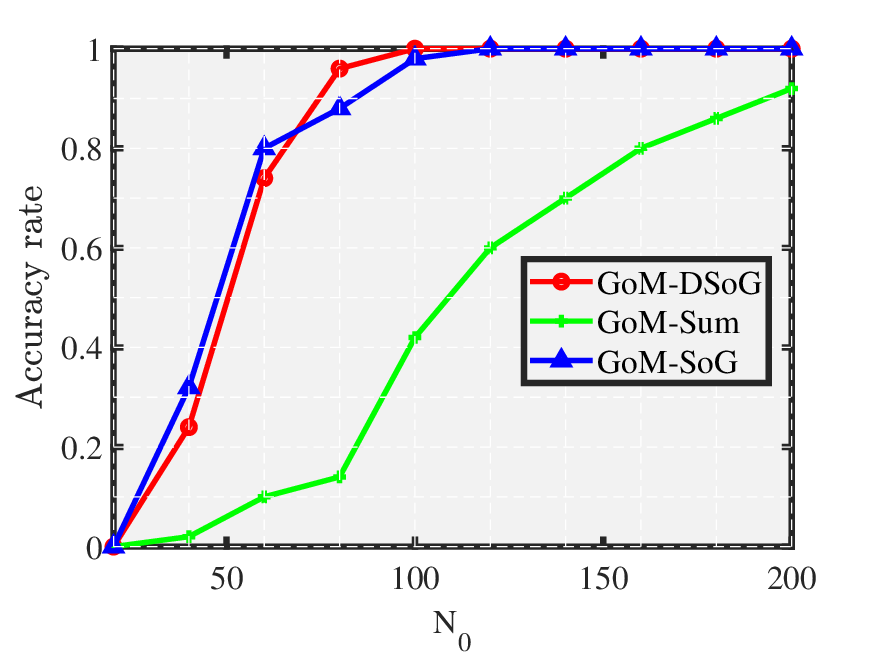}}
\subfigure[$\rho=5$]{\includegraphics[width=0.3\textwidth]{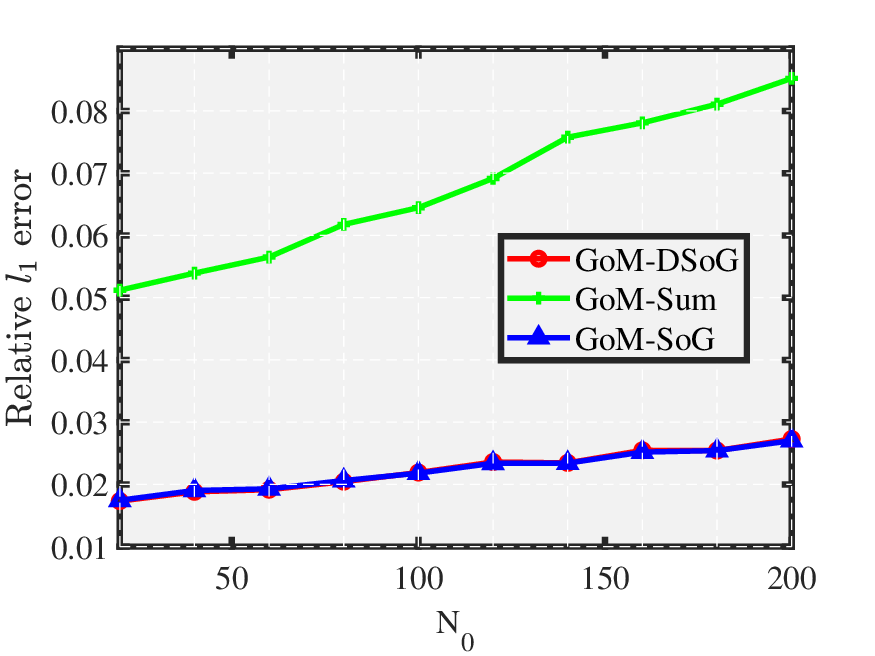}}
\subfigure[$\rho=5$]{\includegraphics[width=0.3\textwidth]{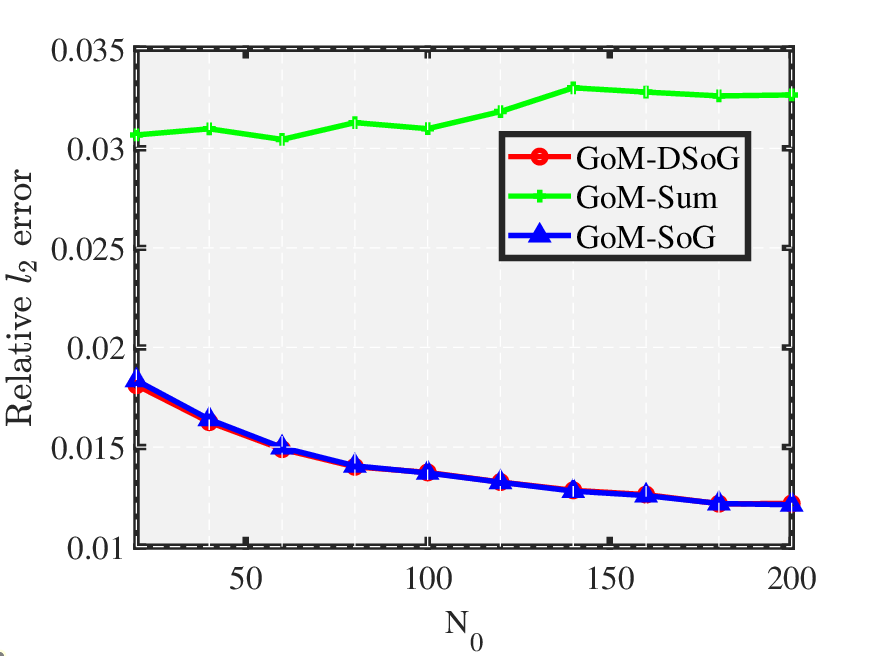}}
\subfigure[$\rho=5$]{\includegraphics[width=0.3\textwidth]{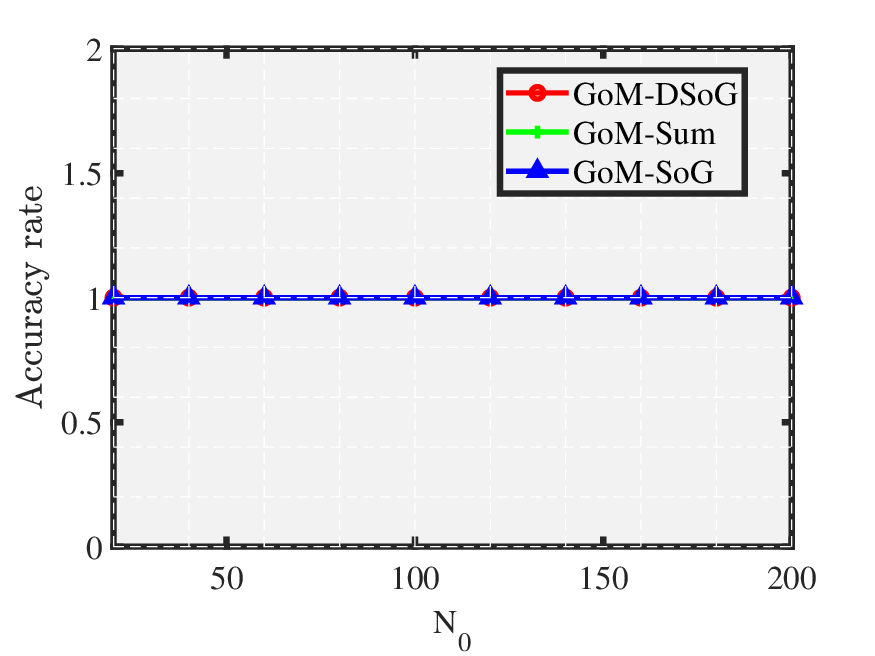}}
%}
\caption{Experiment 4.}
\label{EX4} %% label for entire figure
\end{figure}

\textbf{Experiment 5: Effect of changing $K$.} Set $N=100K, N_{0}=N/K, L=5$. Let $K$ vary in $\{1, 2, \ldots, 8\}$. For the low-intensity case, set $\rho=0.5$. For the high-intensity case, set $\rho=5$. Results are summarized in Figure \ref{EX5}. We observe that (a) GoM-DSoG and GoM-SoG significantly outperform GoM-Sum in the dense case; (b) all methods fail to estimate $K$ when the true $K$ is 1. However, this is not a critical issue since we typically assume that there are at least 2 latent classes for the task of grade of membership analysis.
\begin{figure}
\centering
%\resizebox{\columnwidth}{!}{
\subfigure[$\rho=0.5$]{\includegraphics[width=0.3\textwidth]{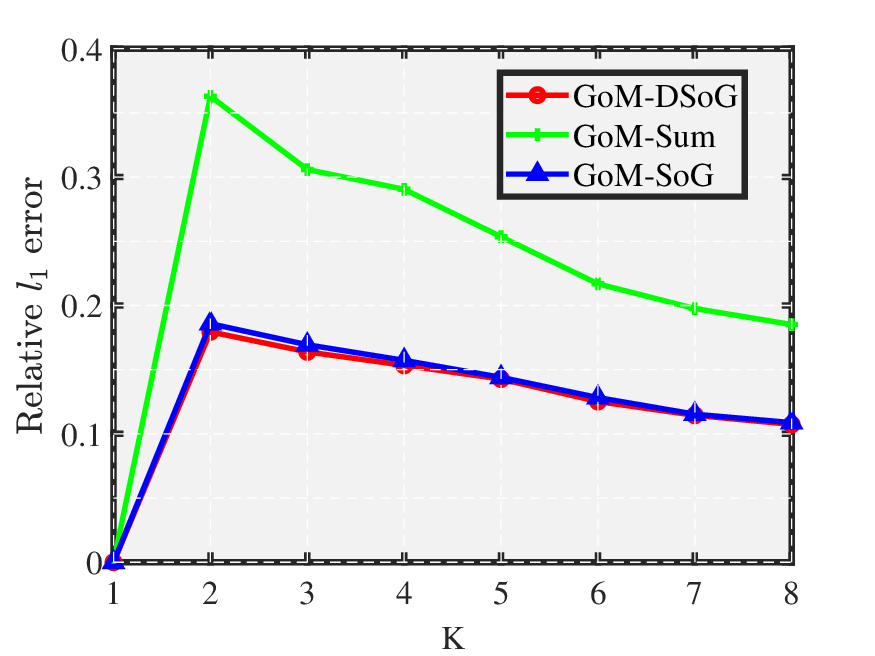}}
\subfigure[$\rho=0.5$]{\includegraphics[width=0.3\textwidth]{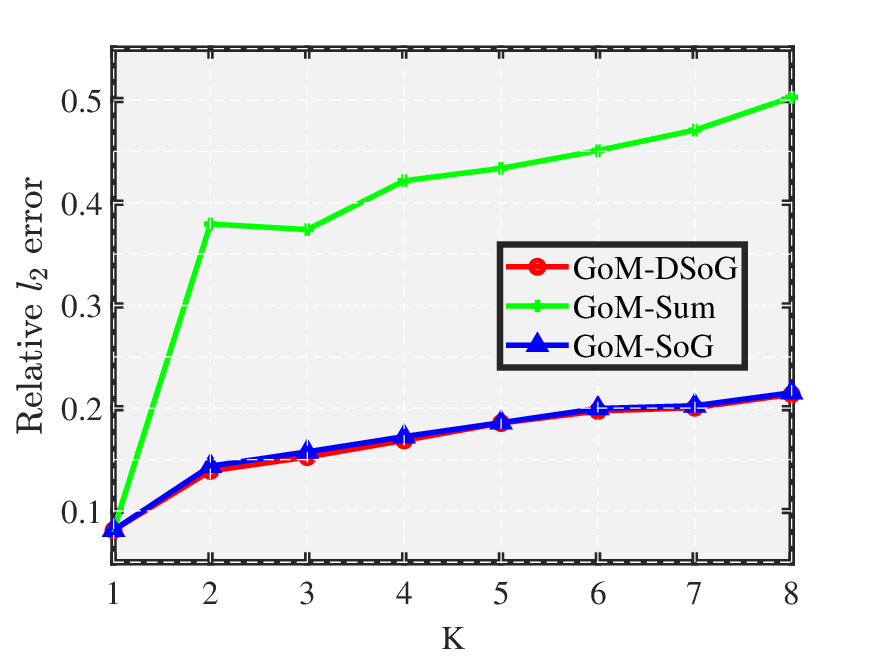}}
\subfigure[$\rho=0.5$]{\includegraphics[width=0.3\textwidth]{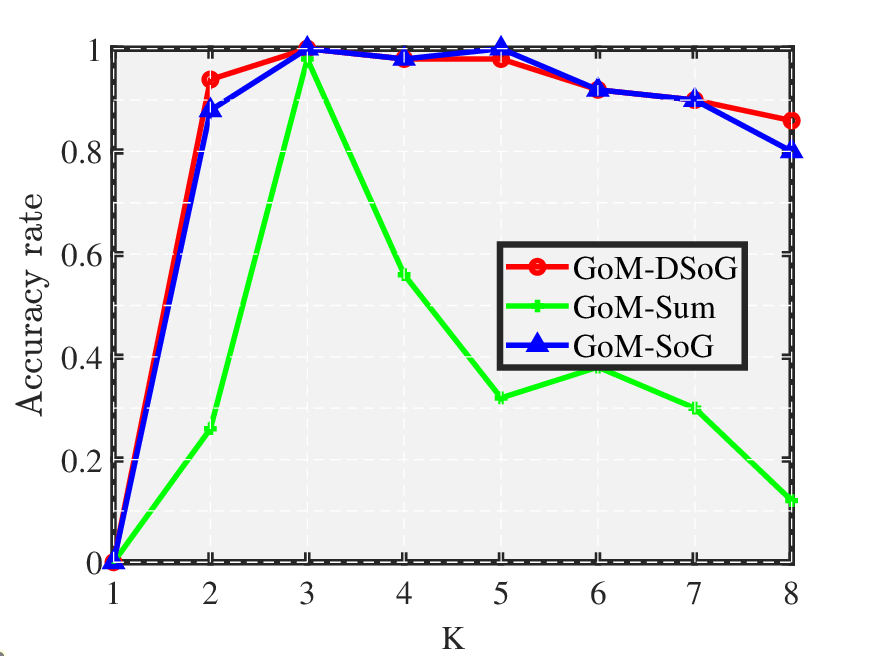}}
\subfigure[$\rho=5$]{\includegraphics[width=0.3\textwidth]{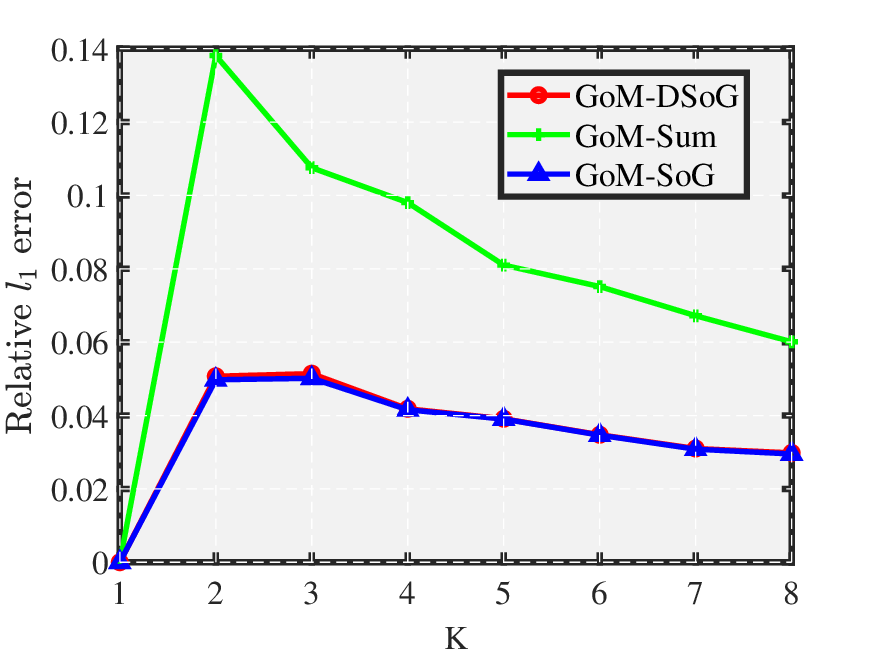}}
\subfigure[$\rho=5$]{\includegraphics[width=0.3\textwidth]{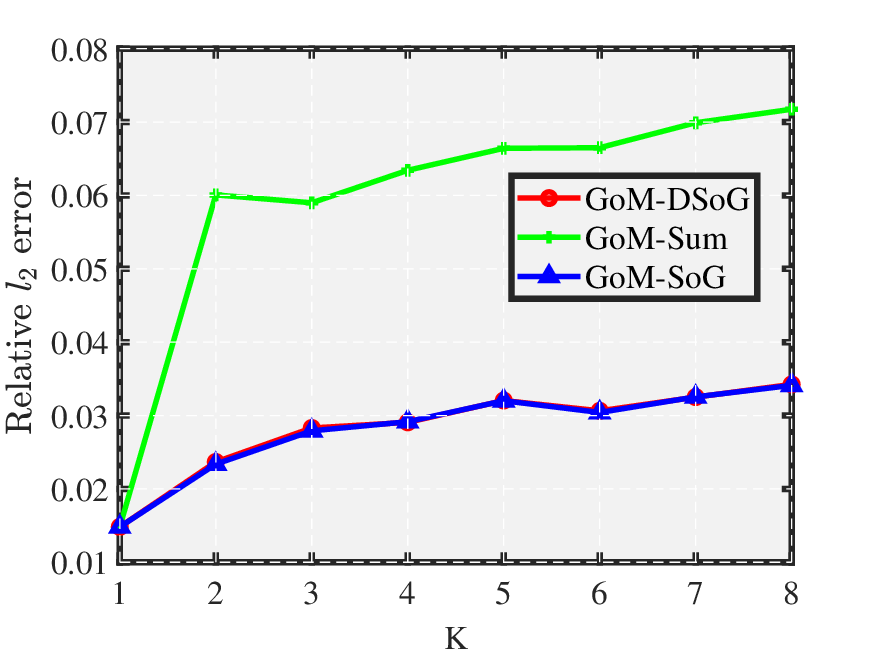}}
\subfigure[$\rho=5$]{\includegraphics[width=0.3\textwidth]{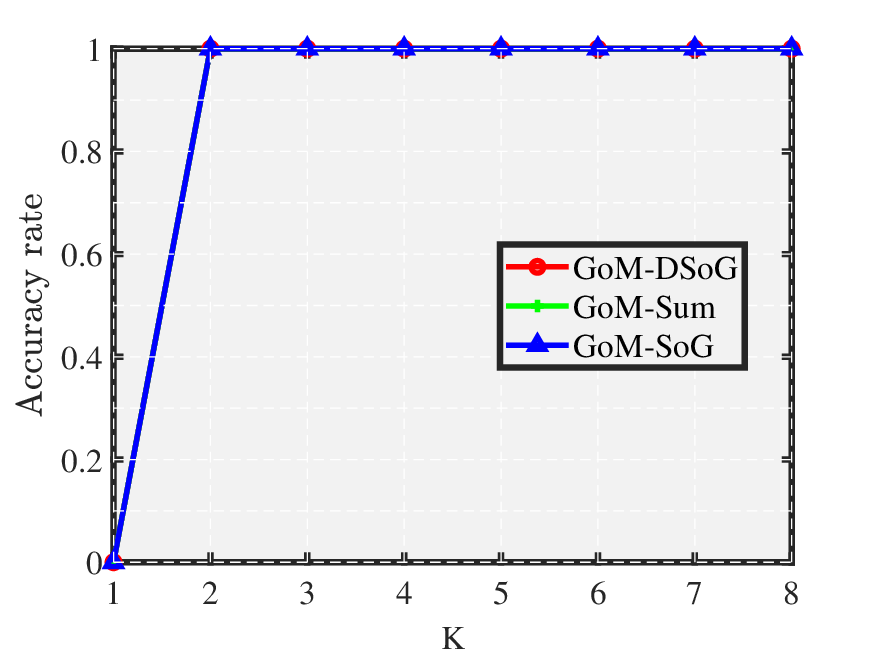}}
%}
\caption{Experiment 5.}
\label{EX5} %% label for entire figure
\end{figure}
\subsection{Real data}
In real-world data applications, the scarcity of readily available multi-layer ordinal categorical datasets poses a challenge for evaluating the performance of our GoM-DSoG algorithm. To address this, we utilize a single-layer ordinal categorical dataset derived from a survey, which we subsequently transform into a multi-layer structure through a randomized selection process. This dataset, sourced from a survey where participants were asked to generate random numbers within specified ranges corresponding to the Big Five Personality Test (hereafter referred to as BFPT), is accessible for download at \url{https://openpsychometrics.org/_rawdata/}. The BFPT dataset comprises 1369 subjects and 50 items, categorized into five personality traits: extraversion (E1-E10), neuroticism (N1-N10), agreeableness (A1-A10), conscientiousness (C1-C10), and openness (O1-O10). The details of each item can be found in Figure \ref{eTheta123}. For this data, responses 0, 1, 2, 3, 4, and 5 represent different levels of agreement with 0 denoting the lowest intensity and 5 denoting the highest intensity. Consequently, we have $R \in \{0,1,2,3,4,5\}^{1369 \times 50}$. BFPT contains a total of 206 0s, constituting approximately 0.3\% of the dataset ($\frac{206}{1369 \times 50}$). To convert the original single-layer BFPT data into a multi-layer ordinal categorical format, we employ a probabilistic approach. For $l\in[L]$, we set $R_{l}(i,j) = R(i,j)$ with a probability of 0.5, and $R_{l}(i,j) = 0$ otherwise. In this way, we generate a synthetic multi-layer dataset that approximates real-world conditions. In this study, we set $L = 3$ layers for our analysis, and we continue to refer to the transformed dataset as BFPT.

We estimate the number of latent classes by maximizing the modularity defined in Equation (\ref{averagedfuzzymodularity}) using our GoM-DSoG algorithm. The inferred value of $K$ is 3, thus we have $\hat{\Pi}\in[0,1]^{1369\times3}$. We call the three estimated latent classes as Class 1, Class 2, and Class 3. We find that $\sum_{i\in[N]}\hat{\Pi}(i,1)=370.0706, \sum_{i\in[N]}\hat{\Pi}(i,2)=323.1968$, and $\sum_{i\in[N]}\hat{\Pi}(i,3)=675.7326$, implying that most subjects tend to belong to Class 3. To further analyze the estimated mixed membership matrix $\hat{\Pi}$, we call subject $i$ strongly pure if $\mathrm{max}_{k\in[K]}\hat{\Pi}(i,k)\geq0.9$, strongly mixed if $\mathrm{max}_{k\in[K]}\hat{\Pi}(i,k)<0.5$, and moderate pure otherwise. We find that there are 19 strongly pure subjects, 322 strongly mixed subjects, and 1028 moderate pure subjects. The estimated mixed memberships of strongly pure subjects and a ternary visualization of the estimated mixed membership matrix $\hat{\Pi}$ are shown in Figure \ref{PiPure}. In the right panel of Figure \ref{PiPure}, the vertices of the triangle represent pure subjects. For any subject, a position closer to one of the triangle's vertices indicates a higher likelihood of being a pure subject. Conversely, a position closer to the center of the triangle suggests a more mixed membership. The ternary diagram clearly shows that most subjects tend to belong to Class 3.
\begin{figure}
\centering
\resizebox{\columnwidth}{!}{
{\includegraphics[width=0.52\textwidth]{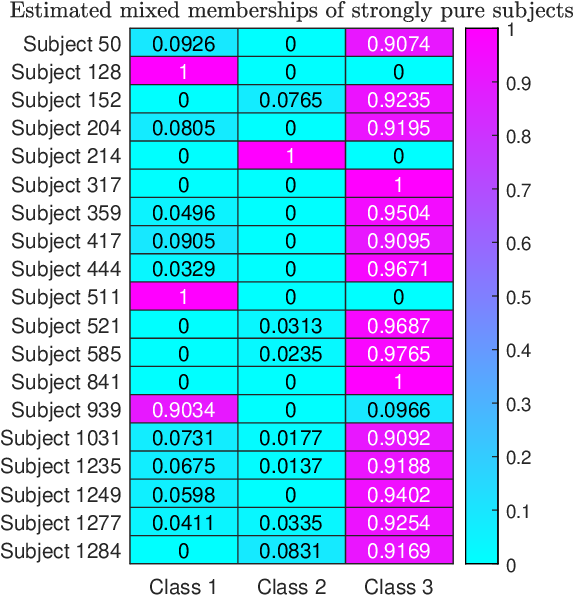}}
{\includegraphics[width=0.73\textwidth]{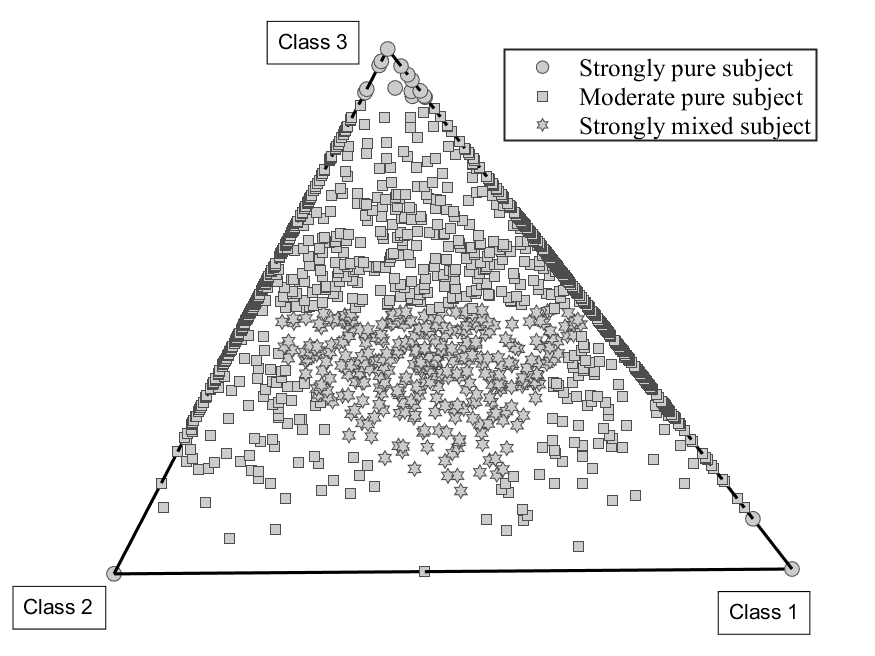}}
}
\caption{Left: estimated mixed memberships of strongly pure subjects for the BFPT data. Right: ternary diagram of $\hat{\Pi}$.}
\label{PiPure} %% label for entire figure
\end{figure}

\begin{figure}
\centering
\resizebox{\columnwidth}{!}{
{\includegraphics[width=6\textwidth]{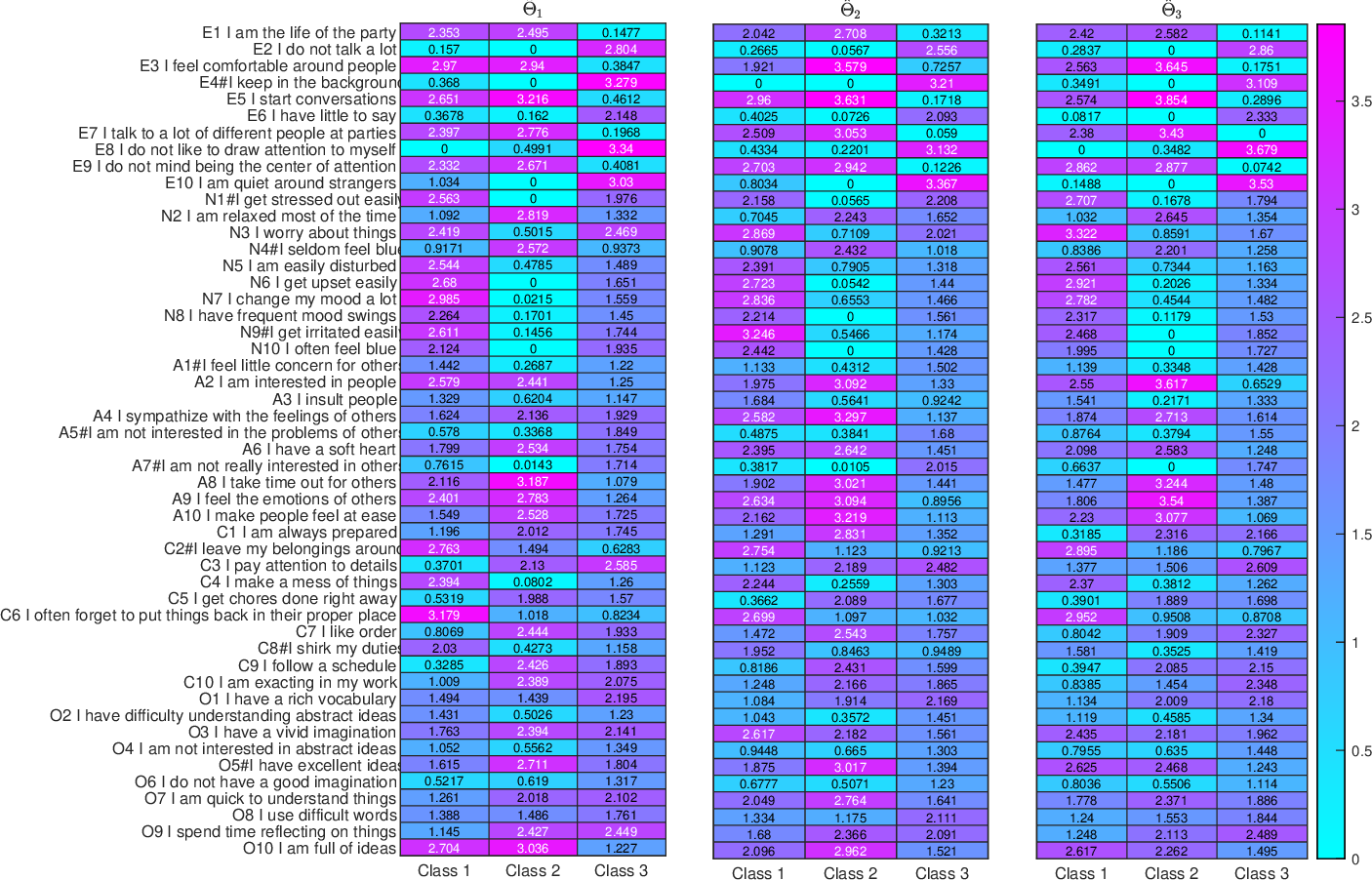}}
}
\caption{Estimated item parameter matrices for the three layers of the BFPT data.}
\label{eTheta123} %% label for entire figure
\end{figure}

Figure \ref{eTheta123} displays the estimated item parameter matrices for the BFPT data. Below, we provide a comprehensive interpretation of these results across the five personality dimensions.

\textbf{Extraversion dimension:}
\begin{itemize}
  \item Class 1: Individuals in this class exhibit moderate agreement with statements like ``I feel comfortable around people" (E3) and ``I talk to a lot of different people at parties" (E7). This reflects a balanced level of extroversion: they can interact with others but do not seek to dominate social interactions.
  \item Class 2: Individuals in this class are more likely to strongly agree with statements such as ``I am the life of the party" (E1) and ``I do not mind being the center of attention" (E9). This suggests they are more outgoing and proactive in social settings, and comfortable engaging with others even in unfamiliar environments.
  \item Class 3: Individuals in this class are less likely to strongly agree with statements such as ``I am the life of the party" (E1) and ``I start conversations" (E5). This indicates a lower tendency toward extroverted behaviors, suggesting they may prefer solitary activities or smaller social gatherings.
\end{itemize}
\textbf{Neuroticism dimension:}
\begin{itemize}
  \item Class 1: Individuals in this class are more likely to strongly agree with statements like ``I get stressed out easily" (N1) and ``I often feel blue" (N10). This indicates a higher likelihood of experiencing stress and negative emotions, suggesting they may be more sensitive to stressors and prone to emotional fluctuations.
  \item Class 2: Individuals in this class are less likely to strongly agree with statements like ``I am relaxed most of the time" (N2) and ``I seldom feel blue" (N4). This indicates lower stress levels and negative emotions, suggesting they are more emotionally stable and better able to cope with stress.
  \item Class 3: Individuals in this class exhibit moderate agreement with statements such as ``I worry about things" (N3) and ``I change my mood a lot" (N7). This suggests they experience some level of worry and mood variability but to a lesser extent than Class 1.
\end{itemize}
\textbf{Agreeableness dimension:}
\begin{itemize}
  \item Class 1: Individuals in this class show moderate agreement with statements for items like ``I am interested in people" (A2) and ``I sympathize with the feelings of others" (A4). This indicates a moderate level of empathy and concern for others.
  \item Class 2: Individuals in this class may have relatively high response values for items such as ``I make people feel at ease" (A10). This suggests they are skilled at making others feel comfortable and are more considerate of others' emotions.
  \item Class 3: Individuals in this class may exhibit relatively low response values for items such as ``I sympathize with the feelings of others" (A4) and ``I feel the emotions of others" (A9). This suggests they may be less considerate of others' feelings.
\end{itemize}
\textbf{Conscientiousness dimension:}
\begin{itemize}
  \item Class 1: Individuals in this class may have relatively low response values for items like ``I am always prepared" (C1) and ``I follow a schedule" (C9). This suggests they may be less organized and less likely to plan ahead.
  \item Class 2: Individuals in this class may have relatively high response values for items like ``I get chores done right away" (C5) and ``I follow a schedule" (C9). This suggests they are more conscientious, tending to complete tasks promptly and to high standards.
  \item Class 3: Individuals in this class exhibit moderate response values for items such as ``I make a mess of things" (C4) and ``I shirk my duties" (C8). This indicates a moderate level of attention to detail and organizational skills.
\end{itemize}
\textbf{Openness dimension:}
\begin{itemize}
  \item Class 1: Individuals in this class show moderate response values. This indicates a moderate level of openness to new experiences and ideas.
  \item Class 2: Individuals in this class may have relatively low response values for items such as ``I have difficulty understanding abstract ideas" (O2), ``I am not interested in abstract ideas" (O4), and ``I do not have a good imagination" (O6). This suggests they are more open to new experiences and ideas, with stronger curiosity and creativity.
  \item Class 3: Individuals in this class may have relatively low response values for items such as ``I have a vivid imagination" (O3) and ``I have excellent ideas" (O5). This suggests they may be less open to new experiences and ideas.
\end{itemize}
\section{Conclusion}\label{sec6}
In this paper, we propose a novel model, the multi-layer GoM model, which extends the traditional GoM model to effectively describe multi-layer ordinal categorical data with a latent mixed membership structure. Our multi-layer GoM assumes that the response matrix for each layer is generated by the classical GoM model, sharing common mixed memberships but with varying item parameter matrices. To facilitate GoM analysis in multi-layer ordinal categorical data, we develop a new approach, GoM-DSoG, inspired by the recently popular technique of debiased spectral clustering in network analysis. We establish the estimation consistency of our method and find that, compared to single-layer analysis, multi-layer analysis exhibits greater power for the task of grade of membership analysis. Furthermore, we propose an efficient approach for estimating the number of latent classes of multi-layer ordinal categorical data with latent mixed membership structures. Our extensive experimental studies demonstrate the effectiveness of our methods in estimating subjects' mixed memberships, the number of latent classes, and other model parameters. To our knowledge, we are the first to explore the grade of membership analysis in multi-layer ordinal categorical data. Our contributions advance the field of ordinal categorical data analysis and have practical implications, particularly in psychological testing and other applications involving multi-layer ordinal categorical data.

The multi-layer GoM framework and the proposed GoM-DSoG algorithm offer a robust and scalable approach for analyzing complex multi-layer ordinal categorical data. However, the current lack of publicly available data underscores a critical methodological gap that our work seeks to address. A promising direction for future work involves applying this methodology to real-world longitudinal studies where the same ordinal-scale instrument is administered repeatedly over time. We encourage researchers in fields such as educational measurement (e.g., tracking student attitudes across a semester with identical surveys), clinical psychology (e.g., monitoring patient symptom severity through repeated clinical assessments), and organizational behavior (e.g., measuring employee engagement across multiple annual surveys) to utilize our model and algorithm to further validate their practical utility and provide valuable insights into the latent mixed memberships of subjects in multi-layer ordinal categorical data.
%%%%%%%%%%%%%%%%%%%%%%%%%%%%%%%%%%%%%%%%%%%%%%%%%%%%%%%%%%%%%%%%%%%%%%%%%%%%%%%%%%%%%%%%%%%%%%%%%%%%%%%%%%%%%%%%%%%%%%%%%%%%
\section*{Acknowledgements}
H.Q. was supported by the Scientific Research Foundation of Chongqing University of Technology (Grant No. 2024ZDR003), and the Science and Technology Research Program of Chongqing Municipal Education Commission (Grant No. KJQN202401168).
\section{Supplementary Materials}\label{s1}
\subsection{Proof of Proposition 1}
\begin{proof}
Since $\mathcal{S}=\sum_{l=1}^{L}\Pi\Theta'_{l}\Theta_{l}\Pi'=\Pi(\sum_{l=1}^{L}\Theta'_{l}\Theta_{l})\Pi'$, and given that $\mathrm{rank}(\sum_{l=1}^{L}\Theta'_{l}\Theta_{l})=K$ and $\Pi$ satisfies the condition stated in Equation (2.2) in the main context, by the first result in Theorem 2.1 \citep{mao2021estimating}, $\Pi$ is identifiable up to a permutation, i.e.,  $\tilde{\Pi}=\Pi\mathcal{P}$. Furthermore, based on the second bullet of Lemma 1 and the fact that $\tilde{\Pi}=\Pi\mathcal{P}$, we have:
\begin{align*}
\Theta_{l} &= \mathscr{R}'_{l}\Pi (\Pi'\Pi)^{-1}, \\
\tilde{\Theta}_{l} &= \mathscr{R}'_{l}\tilde{\Pi} (\tilde{\Pi}'\tilde{\Pi})^{-1} = \mathscr{R}'_{l}\Pi\mathcal{P}(\mathcal{P}'\Pi'\Pi\mathcal{P})^{-1} = \mathscr{R}'_{l}\Pi\mathcal{P}\mathcal{P}^{-1}(\Pi'\Pi)^{-1}(\mathcal{P}')^{-1} = \Theta_{l}\mathcal{P}.
\end{align*}
Thus, for all $l\in[L]$, $\Theta_{l}$ is also identifiable up to the same permutation as $\Pi$.
\end{proof}
\subsection{Proof of Lemma 1}
\begin{proof}
$\mathcal{S}=\sum_{l=1}^{L}\Pi\Theta'_{l}\Theta_{l}\Pi'=\Pi(\sum_{l=1}^{L}\Theta'_{l}\Theta_{l})\Pi'=U\Lambda U'$ gives $U=\Pi(\sum_{l=1}^{L}\Theta'_{l}\Theta_{l})\Pi'U\Lambda^{-1}$. Let $X=(\sum_{l=1}^{L}\Theta'_{l}\Theta_{l})\Pi'U\Lambda^{-1}$. Then, $U=\Pi X$ gives $U(\mathcal{I},:)=(\Pi X)(\mathcal{I},:)=\Pi(\mathcal{I},:)X=X$, which leads to $U=\Pi U(\mathcal{I},:)$. Additionally, since $\mathscr{R}_{l}=\Pi\Theta'_{l}$, we have $\Theta_{l}=\mathscr{R}'_{l}\Pi(\Pi'\Pi)^{-1}$ for $l\in[L]$.
\end{proof}
\subsection{Proof of Theorem 1}
\begin{proof}
First, we prove the following lemma.
\begin{lem}\label{boundSumInfinity}
If Assumption 1 is satisfied, with probability $1-o(\frac{1}{(N+J+L)^{3}})$,
\begin{align*}
\|S-\mathcal{S}\|_{\infty}=O(\sqrt{\rho^{2}NJL\mathrm{log}(N+J+L)})+\rho^{2}JL.
\end{align*}
\end{lem}
\begin{proof}
Under the proposed model, we get
\begin{align*}
 \|S-\mathcal{S}\|_{\infty}&=\mathrm{max}_{i\in[N]}\sum_{h\in[N]}|S(i,h)-\mathcal{S}(i,h)|=\mathrm{max}_{i\in[N]}\sum_{h\in[N]}|(\sum_{l\in[L]}R_{l}R'_{l}-D_{l})(i,h)-(\sum_{l\in[L]}\mathscr{R}_{l}\mathscr{R}'_{l})(i,h)|\\
 &=\mathrm{max}_{i\in[N]}\sum_{h\in[N]}|\sum_{l\in[L]}\sum_{j\in[J]}R_{l}(i,j)R_{l}(h,j)-\sum_{l\in[L]}\sum_{j\in[J]}\mathscr{R}_{l}(i,j)\mathscr{R}_{l}(h,j)-\sum_{l\in[L]}D_{l}(i,h)|\\
 &=\mathrm{max}_{i\in[N]}\sum_{h\neq i, h\in[N]}|\sum_{l\in[L]}\sum_{j\in[J]}R_{l}(i,j)R_{l}(h,j)-\sum_{l\in[L]}\sum_{j\in[J]}\mathscr{R}_{l}(i,j)\mathscr{R}_{l}(h,j)|\\
 &+\mathrm{max}_{i\in[N]}|\sum_{l\in[L]}\sum_{j\in[J]}R^{2}_{l}(i,j)-\sum_{l\in[L]}\sum_{j\in[J]}\mathscr{R}^{2}_{l}(i,j)-\sum_{l\in[L]}D_{l}(i,i)|\\
 &=\mathrm{max}_{i\in[N]}\sum_{h\neq i, h\in[N]}|\sum_{l\in[L]}\sum_{j\in[J]}(R_{l}(i,j)R_{l}(h,j)-\mathscr{R}_{l}(i,j)\mathscr{R}_{l}(h,j))|+\mathrm{max}_{i\in[N]}\sum_{l\in[L]}\sum_{j\in[J]}\mathscr{R}^{2}_{l}(i,j)\\
 &\leq\mathrm{max}_{i\in[N]}\sum_{h\neq i, h\in[N]}|\sum_{l\in[L]}\sum_{j\in[J]}(R_{l}(i,j)R_{l}(h,j)-\mathscr{R}_{l}(i,j)\mathscr{R}_{l}(h,j))|+\sum_{l\in[L]}\sum_{j\in[J]}\rho^{2}\\
 &=\mathrm{max}_{i\in[N]}\sum_{h\neq i, h\in[N]}|\sum_{l\in[L]}\sum_{j\in[J]}(R_{l}(i,j)R_{l}(h,j)-\mathscr{R}_{l}(i,j)\mathscr{R}_{l}(h,j))|+\rho^{2}JL\\
 \end{align*}
Define $v$ as any $(N-1)$-by-$1$ vector. For $i\in[N]$ and $h\neq i, h\in[N]$, let $z_{(ih)}=\sum_{l\in[L]}\sum_{j\in[J]}(R_{l}(i,j)R_{l}(h,j)-\mathscr{R}_{l}(i,j)\mathscr{R}_{l}(h,j))$. Then, for $i\in[N]$, define $F_{(i)}=\sum_{h\neq i,h\in[N]}z_{(ih)}v(h)$. Let $\tau=\mathrm{max}_{i\in[N], h\neq i, h\in[N]}|z_{(ih)}|$. Given that $\mathbb{E}(z_{(ih)})=0$, to simplify our analysis, we assume $\tau$ is no larger than a constant $C$. For $i\in[N]$ and $h\neq i, h\in[N]$, the following conclusions hold:
\begin{itemize}
  \item $\mathbb{E}(z_{(ih)}v(h))=0$ because $h\neq i$.
  \item $|z_{(ih)}v(h)|\leq C\|v\|_{\infty}$.
  \item Set $\sigma^{2}=\sum_{h\neq i, h\in[N]}\mathbb{E}(z^{2}_{(ih)}v^{2}(h))$. Under our multi-layer GoM model, we have
  \begin{align*}
\sum_{h\neq i,h\in[H]}\mathbb{E}[z^{2}_{(ih)}v^{2}(h)]&=\sum_{h\neq i,h\in[N]}v^{2}(h)\mathbb{E}[z^{2}_{(ih)}]\\
&=\sum_{h\neq i,h\in[N]}v^{2}(h)\sum_{l\in[L]}\sum_{j\in[J]}\mathbb{E}((R_{l}(i,j)R_{l}(h,j)-\mathscr{R}_{l}(i,j)\mathscr{R}_{l}(h,j))^{2})\\
&=\sum_{h\neq i,h\in[N]}v^{2}(h)\sum_{l\in[L]}\sum_{j\in[J]}(\mathbb{E}(R^{2}_{l}(i,j))\mathbb{E}(R^{2}_{l}(h,j))-\mathscr{R}^{2}_{l}(i,j)\mathscr{R}^{2}_{l}(h,j))\\
&=\sum_{h\neq i,h\in[N]}v^{2}(h)\sum_{l\in[L]}\sum_{j\in[J]}((\mathrm{Var}(R_{l}(i,j))+\mathscr{R}^{2}_{l}(i,j))(\mathrm{Var}(R_{l}(h,j))\\
&~~+\mathscr{R}^{2}_{l}(h,j))-\mathscr{R}^{2}_{l}(i,j)\mathscr{R}^{2}_{l}(h,j))\\
&=\sum_{h\neq i,h\in[N]}v^{2}(h)\sum_{l\in[L]}\sum_{j\in[J]}(\mathrm{Var}(R_{l}(i,j))\mathrm{Var}(R_{l}(h,j))+\mathscr{R}^{2}_{l}(i,j)\mathrm{Var}(R_{l}(h,j))\\
&~~+\mathscr{R}^{2}_{l}(h,j)\mathrm{Var}(R_{l}(i,j)))\\
&\leq\sum_{h\neq i,h\in[N]}v^{2}(h)\sum_{l\in[L]}\sum_{j\in[J]}(\mathrm{Var}(R_{l}(i,j))\mathrm{Var}(R_{l}(h,j))+\rho^{2}\mathrm{Var}(R_{l}(h,j))\\
&~~+\rho^{2}\mathrm{Var}(R_{l}(i,j)))\\
&=\sum_{h\neq i,h\in[N]}v^{2}(h)\sum_{l\in[L]}\sum_{j\in[J]}(M\frac{\mathscr{R}_{l}(i,j)}{M}(1-\frac{\mathscr{R}_{l}(i,j)}{M})M\frac{\mathscr{R}_{l}(h,j)}{M}(1-\frac{\mathscr{R}_{l}(h,j)}{M})\\
&~~+\rho^{2}M\frac{\mathscr{R}_{l}(h,j)}{M}(1-\frac{\mathscr{R}_{l}(h,j)}{M})+\rho^{2}M\frac{\mathscr{R}_{l}(i,j)}{M}(1-\frac{\mathscr{R}_{l}(i,j)}{M}))\\
&\leq\sum_{h\neq i,h\in[N]}v^{2}(h)\sum_{l\in[L]}\sum_{j\in[J]}(\rho^{2}+2\rho^{3})=(\rho^{2}+2\rho^{3})JL\|v\|^{2}_{F}.
\end{align*}
\end{itemize}
According to the Bernstein inequality in Theorem 1.4  \cite{tropp2012user}, for any $t>0$, we obtain
\begin{align*}
\mathbb{P}(|F_{(i)}|\geq t)\leq\mathrm{exp}(\frac{-t^{2}/2}{(\rho^{2}+2\rho^{3})\|v\|^{2}_{F}JL+\frac{C\|v\|_{\infty}t}{3}}).
\end{align*}
Set $t$ as $\sqrt{(\rho^{2}+2\rho^{3}\|v\|^{2}_{F}JL\mathrm{log}(N+J+L)}\times\frac{\sqrt{\alpha+1}+\sqrt{\alpha+19}}{3}\sqrt{\alpha+1}$ for any $\alpha\geq0$. Assuming that $(\rho^{2}+2\rho^{2})\|v\|^{2}_{F}JL\geq C^{2}\|v\|^{2}_{\infty}\mathrm{log}(N+J+L)$ is satisfied, we get:
\begin{align*}
\mathbb{P}(|F_{(i)}|\geq t)&\leq\mathrm{exp}(-(\alpha+1)\mathrm{log}(N+J+L)\frac{1}{\frac{18}{(\sqrt{\alpha+1}+\sqrt{\alpha+19})^{2}}+\frac{2\sqrt{\alpha+1}}{\sqrt{\alpha+1}+\sqrt{\alpha+19}}\sqrt{\frac{C^{2}\|V\|^{2}_{\infty}\mathrm{log}(N+J+L)}{(\rho^{2}+2\rho^{3})\|v\|^{2}_{F}JL}}})\\
&\leq\frac{1}{(N+J+L)^{\alpha+1}}.
\end{align*}
Setting $v\in\{-1,1\}^{(N-1)\times 1}$ and $\alpha=3$ gives: when $(\rho^{2}+2\rho^{3})(N-1)JL \gg \mathrm{log}(N+J+L)$, which is equivalent to $\rho^{2}NJL \gg \mathrm{log}(N+J+L)$ since we allow $\rho$ to approach zero and $N-1 \approx N$ for large $N$, with probability $1-o(\frac{1}{(N+J+L)^{3}})$,
\begin{align*}
\mathrm{max}_{i\in[N]}|F_{(i)}|=O(\sqrt{\rho^{2}NJL\mathrm{log}(N+J+L)}).
\end{align*}
Hence, we have $\|S-\mathcal{S}\|_{\infty}=O(\sqrt{\rho^{2}NJL\mathrm{log}(N+J+L)})+\rho^{2}JL$.
\end{proof}

Theorem 4.2 \citep{cape2019the} says that when $|\lambda_{K}(\mathcal{S})|\geq 4\|S-\mathcal{S}\|_{\infty}$ is satisfied,
\begin{align*}
\|\hat{U}-U\mathcal{O}\|_{2\rightarrow\infty}\leq14\frac{\|S-\mathcal{S}\|_{\infty}\|U\|_{2\rightarrow\infty}}{|\lambda_{K}(\mathcal{S})|},
\end{align*}
where $\mathcal{O}$ is an orthogonal matrix. Setting $\varpi:=\|\hat{U}\hat{U}'-UU'\|_{2\rightarrow\infty}$ gives
\begin{align*}
\varpi\leq2\|\hat{U}-U\mathcal{O}\|_{2\rightarrow\infty}\leq14\frac{\|S-\mathcal{S}\|_{\infty}\|U\|_{2\rightarrow\infty}}{|\lambda_{K}(\mathcal{S})|}.
\end{align*}
Combining Condition 1 with Lemma 3.1 \citep{mao2021estimating} gives $\|U\|_{2\rightarrow\infty}=O(\sqrt{\frac{1}{N}})$. Thus,
\begin{align*}
\varpi=O(\frac{\|S-\mathcal{S}\|_{\infty}}{|\lambda_{K}(\mathcal{S})|\sqrt{N}}). \end{align*}
For $|\lambda_{K}(\mathcal{S})|$, we have $|\lambda_{K}(\mathcal{S})|\geq\rho^{2}\lambda_{K}(\Pi'\Pi)|\lambda_{K}(\sum_{l\in[L]}B'_{l}B_{l})|$, which gives
\begin{align*}
\varpi=O(\frac{\|S-\mathcal{S}\|_{\infty}}{\rho^{2}\sqrt{N}\lambda_{K}(\Pi'\Pi)|\lambda_{K}(\sum_{l\in[L]}B'_{l}B_{l})|}). \end{align*}
By comparing the steps of our GoM-DSoG algorithm with that of the Algorithm 1 \citep{mao2021estimating} in the task of estimating the mixed membership matrix $\Pi$, we find that all of their steps are the same except that the eigenvector matrix $\hat{U}$ is obtained from $S$ in GoM-DSoG while it is from the adjacency matrix in Algorithm 1 in \citep{mao2021estimating}, where we do not consider the prune step in \citep{mao2021estimating}. Therefore, by Equation (3) of 3.2 \citep{mao2021estimating}, Condition 1, Assumption 2, and Lemma \ref{boundSumInfinity}, we get
\begin{align*}
\mathrm{max}_{i\in[N]}\|e'_{i}(\hat{\Pi}-\Pi\mathcal{P})\|_{1}&=O(\varpi\kappa(\Pi'\Pi)\sqrt{\lambda_{1}(\Pi'\Pi)})=O(\frac{\|S-\tilde{S}\|_{\infty}}{\rho^{2}NJL})\\
&=O(\sqrt{\frac{\mathrm{log}(N+J+L)}{\rho^{2}NJL}})+O(\frac{1}{N}),
\end{align*}
where $\kappa(\cdot)$ denotes the condition number.
\end{proof}
\subsection{Proof of Lemma 2}
\begin{proof}
Let $\Pi(\Pi'\Pi)^{-1}=G$ and $\hat{\Pi}(\hat{\Pi}'\hat{\Pi})^{-1}=\hat{G}$. By Condition 1, we get
\begin{align*}
\|\sum_{l\in[L]}(\hat{\Theta}_{l}-\Theta_{l}\mathcal{P})\|&=\|(\sum_{l\in[L]}R_{l})'\hat{G}-(\sum_{l\in[L]}\mathscr{R}_{l})'G\mathcal{P}\|\leq\|(\sum_{l\in[L]}(R_{l}-\mathscr{R}_{l}))'\hat{G}\|+\|(\sum_{l\in[L]}\mathscr{R}_{l})'(\hat{G}-G\mathcal{P})\|\\
&\leq\|\sum_{l\in[L]}(R_{l}-\mathscr{R}_{l})\|\|\hat{G}\|+\rho \|\Pi\|\|\sum_{l\in[L]}B_{l}\|\|\hat{G}-G\mathcal{P}\|\\
&=O(\frac{\|\sum_{l\in[L]}(R_{l}-\mathscr{R}_{l})\|}{\sqrt{N}})+\rho|\lambda_{1}(\sum_{l\in[L]}B_{l})|.
\end{align*}
According to the first statement of Lemma 2 \citep{MLLCM},
assume that $\rho L\mathrm{max}(N,J)\gg\mathrm{log}(N+J+L)$, with probability $1-o(\frac{1}{(N+J+L)^{3}})$, we have $\|\sum_{l\in[L]}(R_{l}-\mathscr{R}_{l})\|=O(\sqrt{\rho L\mathrm{max}(N,J)\mathrm{log}(N+J+L)})$. When $|\lambda_{1}(\sum_{l\in[L]}B_{l})|\geq c_{2}\sqrt{J}L$, we have
\begin{align*}
\frac{\|\sum_{l\in[L]}(\hat{\Theta}_{l}-\Theta_{l}\mathcal{P})\|_{F}}{\|\sum_{l\in[L]}\Theta_{l}\|_{F}}&\leq\frac{\sqrt{K}\|\sum_{l\in[L]}(\hat{\Theta}_{l}-\Theta_{l}\mathcal{P})\|}{\rho|\lambda_{1}(\sum_{l\in[L]}B_{l})|}=O(\frac{\|\sum_{l\in[L]}(R_{l}-\mathscr{R}_{l})\|}{\rho L\sqrt{NJ}})\\
&=O(\sqrt{\frac{\mathrm{max}(N,J)\mathrm{log}(N+J+L)}{\rho NJL}}).
\end{align*}
\end{proof}

\section{MATLAB codes of GoM-DSoG}\label{MatlabCodes}
The MATLAB codes of GoM-DSoG are provided below:
\begin{lstlisting}
function [Pi_hat,Theta_all_hat]=DSoG(R_all,K)
% An implementation of GoM-DSoG algorithm
% Inputs:
%       R_all: a tensor with R_all (:,:, l) being the N*J matrix of layer
% l, where N and J denote the number subjects and items, respectively
%       K: the number of latent classes
% Outputs:
%       Pi_hat: the N*K estimated membership matrix
%       Theta_all_hat: a tensor with Theta_all_hat(:,:,l) being the J*K
% estimated item parameter matrix.
N=size(R_all,1);J=size(R_all,2);L=size(R_all,3);M=max(max(max(R_all)));
S=zeros(N,N);
for l=1:L
    Rl=R_all(:,:,l);Dl=sum(Rl.^2,2);Dl=diag(Dl);S=S+Rl*Rl'-Dl;
end
[U,~]=eigs(S,K);[ pure ]=SPA(U,K);C=U(pure,:);Pi_hat=U/C;
Pi_hat=max(0,Pi_hat);
for i=1:N
    Pi_hat(i,:)=Pi_hat(i,:)/sum(Pi_hat(i,:));
end
Theta_all_hat=zeros(J,K,L);
for l=1:L
    Rl=R_all(:,:,l);
    theta_l_hat=min(M,max(0,Rl'*Pi_hat*inv(Pi_hat'*Pi_hat)));
    Theta_all_hat(:,:,l)=theta_l_hat;
end
\end{lstlisting}
\begin{lstlisting}
function [ pure ]=SPA(X,K)
% An implementation of SPA algorithm
% Input: X: n by K data matrix (eigenvectors in our setting),where
% n is the number of nodes, K is the vertices
% Output: pure: set of pure subjects indices
pure = [];
row_norm = vecnorm(X').^2';
for i = 1:K
    [~,idx_tmp] = max(row_norm);pure = [pure idx_tmp];
    U(i,:) = X(idx_tmp,:);
    for j = 1 : i-1
        U(i,:) = U(i,:) - U(i,:)*(U(j,:)'*U(j,:));
    end
    U(i,:) = U(i,:)/norm(U(i,:)); row_norm = row_norm - (X*U(i,:)').^2;
end
end
\end{lstlisting}
%%%%%%%%%%%%%%%%%%%%%%%%%%%%%%%%%%%%%%%%%%%%%%%%%%%%%%%%%%%%%%%%%%%%%%%%%%%%%%%%%%%%%%%%%%
\bibhang=1.7pc
\bibsep=2pt
\fontsize{9}{14pt plus.8pt minus .6pt}\selectfont
\renewcommand\bibname{\large \bf References}
%\begin{thebibliography}{11}
\expandafter\ifx\csname
natexlab\endcsname\relax\def\natexlab#1{#1}\fi
\expandafter\ifx\csname url\endcsname\relax
\def\url#1{\texttt{#1}}\fi
\expandafter\ifx\csname urlprefix\endcsname\relax\def\urlprefix{URL}\fi

  \bibliographystyle{chicago}      % Chicago style, author-year citations
  \bibliography{refMLGoM}

\begin{thebibliography}{}

\bibitem[\protect\citeauthoryear{Agresti}{Agresti}{2012}]{agresti2012categorical}
Agresti, A. (2012).
\newblock {\em Categorical data analysis}, Volume 792.
\newblock John Wiley \& Sons.

\bibitem[\protect\citeauthoryear{Agterberg and Zhang}{Agterberg and
  Zhang}{2024}]{agterberg2024estimating}
Agterberg, J. and A.~R. Zhang (2024).
\newblock Estimating higher-order mixed memberships via the $l_{2,\infty}$
  tensor perturbation bound.
\newblock {\em Journal of the American Statistical Association\/}, 1--11.

\bibitem[\protect\citeauthoryear{Ara{\'u}jo, Saldanha, Galvao, Yoneyama, Chame,
  and Visani}{Ara{\'u}jo et~al.}{2001}]{araujo2001successive}
Ara{\'u}jo, M. C.~U., T.~C.~B. Saldanha, R.~K.~H. Galvao, T.~Yoneyama, H.~C.
  Chame, and V.~Visani (2001).
\newblock The successive projections algorithm for variable selection in
  spectroscopic multicomponent analysis.
\newblock {\em Chemometrics and Intelligent Laboratory Systems\/}~{\em
  57\/}(2), 65--73.

\bibitem[\protect\citeauthoryear{{Cape}, {Tang}, and {Priebe}}{{Cape}
  et~al.}{2019}]{cape2019the}
{Cape}, J., M.~{Tang}, and C.~E. {Priebe} (2019).
\newblock The two-to-infinity norm and singular subspace geometry with
  applications to high-dimensional statistics.
\newblock {\em Annals of Statistics\/}~{\em 47\/}(5), 2405--2439.

\bibitem[\protect\citeauthoryear{Chen and Gu}{Chen and
  Gu}{2024}]{chen2024spectral}
Chen, L. and Y.~Gu (2024).
\newblock A spectral method for identifiable grade of membership analysis with
  binary responses.
\newblock {\em Psychometrika\/}, 1--32.

\bibitem[\protect\citeauthoryear{Chen, Li, and Zhang}{Chen
  et~al.}{2019}]{chen2019joint}
Chen, Y., X.~Li, and S.~Zhang (2019).
\newblock Joint maximum likelihood estimation for high-dimensional exploratory
  item factor analysis.
\newblock {\em Psychometrika\/}~{\em 84}, 124--146.

\bibitem[\protect\citeauthoryear{Erosheva, Fienberg, and Joutard}{Erosheva
  et~al.}{2007}]{erosheva2007describing}
Erosheva, E.~A., S.~E. Fienberg, and C.~Joutard (2007).
\newblock Describing disability through individual-level mixture models for
  multivariate binary data.
\newblock {\em Annals of Applied Statistics\/}~{\em 1\/}(2), 346.

\bibitem[\protect\citeauthoryear{Gillis and Vavasis}{Gillis and
  Vavasis}{2013}]{gillis2013fast}
Gillis, N. and S.~A. Vavasis (2013).
\newblock Fast and robust recursive algorithmsfor separable nonnegative matrix
  factorization.
\newblock {\em IEEE Transactions on Pattern Analysis and Machine
  Intelligence\/}~{\em 36\/}(4), 698--714.

\bibitem[\protect\citeauthoryear{Gillis and Vavasis}{Gillis and
  Vavasis}{2015}]{gillis2015semidefinite}
Gillis, N. and S.~A. Vavasis (2015).
\newblock Semidefinite programming based preconditioning for more robust
  near-separable nonnegative matrix factorization.
\newblock {\em SIAM Journal on Optimization\/}~{\em 25\/}(1), 677--698.

\bibitem[\protect\citeauthoryear{Goodman}{Goodman}{1974}]{goodman1974exploratory}
Goodman, L.~A. (1974).
\newblock Exploratory latent structure analysis using both identifiable and
  unidentifiable models.
\newblock {\em Biometrika\/}~{\em 61\/}(2), 215--231.

\bibitem[\protect\citeauthoryear{Gormley and Murphy}{Gormley and
  Murphy}{2009}]{gormley2009grade}
Gormley, I.~C. and T.~B. Murphy (2009).
\newblock {A grade of membership model for rank data}.
\newblock {\em Bayesian Analysis\/}~{\em 4\/}(2), 265 -- 295.

\bibitem[\protect\citeauthoryear{Gu, Erosheva, Xu, and Dunson}{Gu
  et~al.}{2023}]{gu2023dimension}
Gu, Y., E.~A. Erosheva, G.~Xu, and D.~B. Dunson (2023).
\newblock Dimension-grouped mixed membership models for multivariate
  categorical data.
\newblock {\em Journal of Machine Learning Research\/}~{\em 24\/}(88), 1--49.

\bibitem[\protect\citeauthoryear{Hagenaars and McCutcheon}{Hagenaars and
  McCutcheon}{2002}]{hagenaars2002applied}
Hagenaars, J.~A. and A.~L. McCutcheon (2002).
\newblock {\em Applied latent class analysis}.
\newblock Cambridge University Press.

\bibitem[\protect\citeauthoryear{Jin, Ke, and Luo}{Jin
  et~al.}{2024}]{jin2024mixed}
Jin, J., Z.~T. Ke, and S.~Luo (2024).
\newblock Mixed membership estimation for social networks.
\newblock {\em Journal of Econometrics\/}~{\em 239\/}(2), 105369.

\bibitem[\protect\citeauthoryear{Ke and Wang}{Ke and Wang}{2024}]{ke2024using}
Ke, Z.~T. and M.~Wang (2024).
\newblock Using svd for topic modeling.
\newblock {\em Journal of the American Statistical Association\/}~{\em
  119\/}(545), 434--449.

\bibitem[\protect\citeauthoryear{Klopp, Panov, Sigalla, and Tsybakov}{Klopp
  et~al.}{2023}]{klopp2023assigning}
Klopp, O., M.~Panov, S.~Sigalla, and A.~B. Tsybakov (2023).
\newblock Assigning topics to documents by successive projections.
\newblock {\em Annals of Statistics\/}~{\em 51\/}(5), 1989--2014.

\bibitem[\protect\citeauthoryear{Lei, Chen, and Lynch}{Lei
  et~al.}{2020}]{lei2020consistent}
Lei, J., K.~Chen, and B.~Lynch (2020).
\newblock Consistent community detection in multi-layer network data.
\newblock {\em Biometrika\/}~{\em 107\/}(1), 61--73.

\bibitem[\protect\citeauthoryear{Lei and Lin}{Lei and Lin}{2023}]{lei2023bias}
Lei, J. and K.~Z. Lin (2023).
\newblock Bias-adjusted spectral clustering in multi-layer stochastic block
  models.
\newblock {\em Journal of the American Statistical Association\/}~{\em
  118\/}(544), 2433--2445.

\bibitem[\protect\citeauthoryear{Lin and Lei}{Lin and
  Lei}{2024}]{lin2024dynamic}
Lin, K.~Z. and J.~Lei (2024).
\newblock Dynamic clustering for heterophilic stochastic block models with
  time-varying node memberships.
\newblock {\em arXiv preprint arXiv:2403.05654\/}.

\bibitem[\protect\citeauthoryear{Mao, Sarkar, and Chakrabarti}{Mao
  et~al.}{2021}]{mao2021estimating}
Mao, X., P.~Sarkar, and D.~Chakrabarti (2021).
\newblock Estimating mixed memberships with sharp eigenvector deviations.
\newblock {\em Journal of the American Statistical Association\/}~{\em
  116\/}(536), 1928--1940.

\bibitem[\protect\citeauthoryear{Nepusz, Petr{\'o}czi, N{\'e}gyessy, and
  Bazs{\'o}}{Nepusz et~al.}{2008}]{nepusz2008fuzzy}
Nepusz, T., A.~Petr{\'o}czi, L.~N{\'e}gyessy, and F.~Bazs{\'o} (2008).
\newblock Fuzzy communities and the concept of bridgeness in complex networks.
\newblock {\em Physical Review E\/}~{\em 77\/}(1), 016107.

\bibitem[\protect\citeauthoryear{Newman}{Newman}{2006}]{newman2006modularity}
Newman, M.~E. (2006).
\newblock Modularity and community structure in networks.
\newblock {\em Proceedings of the National Academy of Sciences\/}~{\em
  103\/}(23), 8577--8582.

\bibitem[\protect\citeauthoryear{Newman and Girvan}{Newman and
  Girvan}{2004}]{newman2004finding}
Newman, M.~E. and M.~Girvan (2004).
\newblock Finding and evaluating community structure in networks.
\newblock {\em Physical Review E\/}~{\em 69\/}(2), 026113.

\bibitem[\protect\citeauthoryear{Nylund-Gibson and Choi}{Nylund-Gibson and
  Choi}{2018}]{nylund2018ten}
Nylund-Gibson, K. and A.~Y. Choi (2018).
\newblock Ten frequently asked questions about latent class analysis.
\newblock {\em Translational Issues in Psychological Science\/}~{\em 4\/}(4),
  440.

\bibitem[\protect\citeauthoryear{Paul and Chen}{Paul and
  Chen}{2020}]{paul2020spectral}
Paul, S. and Y.~Chen (2020).
\newblock {Spectral and matrix factorization methods for consistent community
  detection in multi-layer networks}.
\newblock {\em Annals of Statistics\/}~{\em 48\/}(1), 230 -- 250.

\bibitem[\protect\citeauthoryear{Paul and Chen}{Paul and
  Chen}{2021}]{paul2021null}
Paul, S. and Y.~Chen (2021).
\newblock Null models and community detection in multi-layer networks.
\newblock {\em Sankhya A\/}, 1--55.

\bibitem[\protect\citeauthoryear{Pensky and Zhang}{Pensky and
  Zhang}{2019}]{pensky2019spectral}
Pensky, M. and T.~Zhang (2019).
\newblock {Spectral clustering in the dynamic stochastic block model}.
\newblock {\em Electronic Journal of Statistics\/}~{\em 13\/}(1), 678 -- 709.

\bibitem[\protect\citeauthoryear{Qing}{Qing}{2024a}]{qing2024finding}
Qing, H. (2024a).
\newblock Finding mixed memberships in categorical data.
\newblock {\em Information Sciences\/}, 120785.

\bibitem[\protect\citeauthoryear{Qing}{Qing}{2024b}]{MLLCM}
Qing, H. (2024b).
\newblock Latent class analysis for multi-layer categorical data.
\newblock {\em arXiv preprint arXiv:2408.05535\/}.

\bibitem[\protect\citeauthoryear{Qing}{Qing}{2025a}]{qing2025community}
Qing, H. (2025a).
\newblock Community detection by spectral methods in multi-layer networks.
\newblock {\em Applied Soft Computing\/}, 112769.

\bibitem[\protect\citeauthoryear{Qing}{Qing}{2025b}]{qing2025discovering}
Qing, H. (2025b).
\newblock Discovering overlapping communities in multi-layer directed networks.
\newblock {\em Chaos, Solitons $\&$ Fractals\/}~{\em 194}, 116175.

\bibitem[\protect\citeauthoryear{Qing}{Qing}{2025c}]{qing2025mixed}
Qing, H. (2025c).
\newblock Mixed membership estimation for categorical data with weighted
  responses.
\newblock {\em TEST\/}, 1--48.

\bibitem[\protect\citeauthoryear{Qing and Wang}{Qing and
  Wang}{2023}]{qing2023community}
Qing, H. and J.~Wang (2023).
\newblock Community detection for weighted bipartite networks.
\newblock {\em Knowledge-Based Systems\/}~{\em 274}, 110643.

\bibitem[\protect\citeauthoryear{Qing and Wang}{Qing and
  Wang}{2024}]{qing2024bipartite}
Qing, H. and J.~Wang (2024).
\newblock Bipartite mixed membership distribution-free model. a novel model for
  community detection in overlapping bipartite weighted networks.
\newblock {\em Expert Systems with Applications\/}~{\em 235}, 121088.

\bibitem[\protect\citeauthoryear{Robitzsch}{Robitzsch}{2023}]{sirt_3.13-194}
Robitzsch, A. (2023).
\newblock {\em sirt: Supplementary Item Response Theory Models}.
\newblock R package version 3.13-228.

\bibitem[\protect\citeauthoryear{Shang, Erosheva, and Xu}{Shang
  et~al.}{2021}]{shang2021partial}
Shang, Z., E.~A. Erosheva, and G.~Xu (2021).
\newblock Partial-mastery cognitive diagnosis models.
\newblock {\em Annals of Applied Statistics\/}~{\em 15\/}(3), 1529--1555.

\bibitem[\protect\citeauthoryear{Sloane and Morgan}{Sloane and
  Morgan}{1996}]{sloane1996introduction}
Sloane, D. and S.~P. Morgan (1996).
\newblock An introduction to categorical data analysis.
\newblock {\em Annual Review of Sociology\/}~{\em 22\/}(1), 351--375.

\bibitem[\protect\citeauthoryear{Su, Guo, Chang, and Yang}{Su
  et~al.}{2024}]{su2024spectral}
Su, W., X.~Guo, X.~Chang, and Y.~Yang (2024).
\newblock Spectral co-clustering in multi-layer directed networks.
\newblock {\em Computational Statistics $\&$ Data Analysis\/}, 107987.

\bibitem[\protect\citeauthoryear{Tropp}{Tropp}{2012}]{tropp2012user}
Tropp, J.~A. (2012).
\newblock User-friendly tail bounds for sums of random matrices.
\newblock {\em Foundations of Computational Mathematics\/}~{\em 12}, 389--434.

\bibitem[\protect\citeauthoryear{Woodbury, Clive, and Garson~Jr}{Woodbury
  et~al.}{1978}]{woodbury1978mathematical}
Woodbury, M.~A., J.~Clive, and A.~Garson~Jr (1978).
\newblock Mathematical typology: a grade of membership technique for obtaining
  disease definition.
\newblock {\em Computers and Biomedical Research\/}~{\em 11\/}(3), 277--298.

\bibitem[\protect\citeauthoryear{Xu, Zhen, and Wang}{Xu
  et~al.}{2023}]{xu2023covariate}
Xu, S., Y.~Zhen, and J.~Wang (2023).
\newblock Covariate-assisted community detection in multi-layer networks.
\newblock {\em Journal of Business $\&$ Economic Statistics\/}~{\em 41\/}(3),
  915--926.

\end{thebibliography}

\end{document}